\documentclass[reqno,12pt]{amsart}
\usepackage{amsmath,amssymb,latexsym,amsthm,mathrsfs,mathtools,xcolor,wasysym}
\usepackage[a4paper,hmargin={2.54cm,2.54cm},vmargin={3.17cm,3.17cm}]{geometry}
\usepackage[hidelinks, colorlinks=true, linkcolor=black, citecolor=red, anchorcolor=green]{hyperref}

\theoremstyle{definition}
\newtheorem{theorem}{Theorem}[section]
\newtheorem{lemma}[theorem]{Lemma}
\newtheorem{corollary}[theorem]{Corollary}
\newtheorem{remark}[theorem]{Remark}
\newtheorem{proposition}[theorem]{Proposition}
\newtheorem{definition}[theorem]{Definition}
\newtheorem*{hypo1}{Hypothesis H1}
\newtheorem*{hypo2}{Hypothesis H2}
\numberwithin{equation}{section}

\newcommand{\M}{\mathcal M}
\newcommand{\E}{\mathcal E}
\newcommand{\Id}{\text{Id}}
\newcommand{\lt}{\left}
\newcommand{\rt}{\right}
\newcommand{\mi}{\mathrm{i}}
\newcommand{\md}{\mathrm{d}}
\newcommand{\me}{\mathrm{e}}
\newcommand{\nf}{\ensuremath{\mathcal{NF}}}

\title[Reducibility of Quantum Harmonic on $\mathbb R^d$]{Reducibility of Quantum Harmonic Oscillator on $\mathbb{R}^d$ Perturbed by a Quasi - periodic Potential with Logarithmic Decay
}
\author{\bf Zhenguo  Liang$^1$ $\cdot$ Zhiqiang Wang$^1$}

\thanks{\XBox\quad Zhenguo Liang\\ \indent\qquad zgliang@fudan.edu.cn\\[4pt] \indent\qquad Zhiqiang Wang\\ 
\indent\qquad 19110180010@fudan.edu.cn\\[4pt]
\indent $^1$\hspace{4ex}School of Mathematical Sciences and
Key Lab of Mathematics for Nonlinear Science, \\
\indent\qquad Fudan University,
Shanghai 200433, People's Republic of China}
\date{}
\begin{document}
\maketitle

\begin{abstract} We prove  the reducibility of quantum harmonic oscillators in $\mathbb R^d$ perturbed by a quasi-periodic in time potential $V(x,\omega t)$ with \textsl{logarithmic decay}.  By a new estimate  built for solving the homological equation  we improve the reducibility result by 
Gr\'ebert-Paturel(Annales de la Facult\'e des sciences de Toulouse : Math\'ematiques. \textbf{28}, 2019).  
\end{abstract}
\section{Introduction and Main Results}
\subsection{Statement of the results}
In this paper we consider the linear equation
\begin{equation}\label{maineqn}
\mi\partial_tu=(-\Delta+|x|^2)u+\epsilon V(x,\omega t)u,\quad u=u(t,x),~t\in\mathbb{R},x\in\mathbb{R}^d,
\end{equation}
where $\epsilon\ge0$ is a small parameter and frequency vector $\omega$ of the forced oscillations is regarded as a parameter in $\mathcal D_0:=[0,2\pi)^n$. We assume that the potential $V:\mathbb{R}^d\times\mathbb{T}^n\ni(x,\varphi)\mapsto V(x,\varphi)\in\mathbb{R}$ is continuous in all its variables and analytic in $\varphi$, where $\mathbb{T}^n=\mathbb{R}^n/2\pi\mathbb{Z}^n$ denotes the $n$-dimensional torus. For  $\sigma>0$, the function $V(x,\varphi)$ analytically in $\varphi$ extends to the strip
$
\mathbb{T}^n_\sigma=\{(a+b\mi)\in\mathbb{C}^n/2\pi\mathbb{Z}^n:|b|<\sigma\}
$
 and for all $(x,\varphi)\in\mathbb{R}^d\times\mathbb{T}_\sigma^n$ verifies
\begin{equation}\label{potentialcd}
|V(x,\varphi)|\le C(1+\ln(1+|x|^2))^{-2\iota},
\end{equation}
where $\iota\ge0$ and $C>0$. 

Before state the main results,  we need some notations.  For readers' convenience we will follow the notations in \cite{GP2019}.  Let
$
H_0=-\Delta+|x|^2=-\Delta+x_1^2+x_2^2+\cdots+x_d^2
$
be the $d$-dimensional quantum harmonic oscillator. Its spectrum is the sum of $d$ copies of the odd integers set, i.e. the spectrum of $H_0$ equals
$
\widehat{\mathcal E}:=\{d,d+2,d+4,\cdots\}.
$
Denote by $E_j$ for $j\in\widehat{\mathcal E}$ the associated eigenspace whose dimension equals
\[
\#\{(i_1,i_2,\cdots,i_d)\in(2\mathbb{N}+1)^d:i_1+i_2+\cdots+i_d=j\}:=d_j\le j^{d-1},\text{ where }\mathbb{N}=\{0,1,2,\cdots\}.
\]
Denote by $\{\Phi_{j,l},l=1,2,\cdots,d_j\}$ the basis of $E_j$ obtained by $d$-tensor product of Hermite functions: $\Phi_{j,l}=\varphi_{i_1}\otimes\varphi_{i_2}\otimes\cdots\otimes\varphi_{i_d}$ for some choice of $i_1+i_2+\cdots+i_d=j$. Then setting
\[
\mathcal{E}:=\{(j,l)\in\widehat{\mathcal{E}}\times\mathbb{N}:l=1,2,\cdots,d_j\}
\text{ and }w_{j,l}=j\text{ for }(j,l)\in\mathcal{E}
\]
$(\Phi_a)_{a\in\mathcal{E}}$ forms the Hermite basis of $L^2(\mathbb{R}^d)$, verifying
$
H_0\Phi_a=w_a\Phi_a,~a\in\mathcal{E}.
$
We define on $\mathcal{E}$ an equivalent relation:
$
a\sim b\iff w_a=w_b
$
and denote by $[a]$ the equivalence class corresponding with $a\in\mathcal{E}$. We notice for later use that $\#[a]\le w_a^{d-1}$ and that abbreviate the eigenspace $E_{w_a}$ as $E_{[a]}$.\\
\indent For $s\ge0$ denote by $\mathcal{H}^s$ the form domain of $H_0^s$ and the domain of $H_0^{s/2}$ endowed by the graph norm. For negative $s$, the space $\mathcal{H}^s$ is the dual of $\mathcal{H}^{-s}$. In particular, for $s\ge0$ an integer we have
\[
\mathcal{H}^s=\{f\in L^2(\mathbb{R}^d):x^\alpha\partial^\beta f\in L^2(\mathbb{R}^d),~\forall\,\alpha,\beta\in\mathbb{N}^{d},|\alpha|+|\beta|\le s\}.
\]
For Hilbert spaces $\mathcal H_1$ and $\mathcal H_2$, we will denote by $\mathcal{B}(\mathcal{H}_1,\mathcal{H}_2)$ the space of bounded linear operators from $\mathcal{H}_1$ to $\mathcal{H}_2$ and write $\mathcal{B}(\mathcal{H}_1,\mathcal{H}_1) $ as $\mathcal{B}(\mathcal{H}_1)$ for simplicity.\\
\indent To a function $u\in\mathcal{H}^s$ we associate the sequence $\xi$ of its Hermite coefficients by the formula $u(x)=\sum_{a\in\mathcal{E}}\xi_a\Phi_a(x)$ and define 
$
\ell_s^2:=\{(\xi_a)_{a\in\mathcal{E}}:\sum_{a\in\mathcal{E}}w_a^s|\xi_a|^2<\infty\}. 
$
Next we will identify $\mathcal{H}^s$ with $\ell_s^2$ by endowing both spaces with the norm 
\[
\|u\|_{\mathcal H^s}=\|\xi\|_s=\left(\sum_{a\in\mathcal{E}}w_a^s|\xi_a|^2\right)^\frac12.
\]

Our main theorem is the following. 
\begin{theorem}\label{mainthm}
Assume that $V$ satisfies \eqref{potentialcd} with $\iota\ge\frac{n+d+1}2$. There exists $\epsilon_*>0$ such that for all $0\le\epsilon<\epsilon_*$ there exists $\mathcal{D}_\epsilon\subset\mathcal{D}_0:=[0,2\pi)^n$ such that for all $\omega\in\mathcal{D}_\epsilon$, the linear Schr\"odinger equation \eqref{maineqn} 
reduces to a linear equation with constant coefficients in $\mathcal{H}^1$.

More precisely, for any $\omega\in\mathcal{D}_\epsilon$, there exists a linear isomorphism $\Psi(\varphi)=\Psi_{\omega,\epsilon}(\varphi)\in\mathcal{B}(\mathcal{H}^p)$, for $p\in[0,1]$, unitary on $L^2(\mathbb{R}^d)$, which analytically depends on $\varphi\in\mathbb{T}^n_{\sigma/2}$ and a bounded Hermitian operator $\mathcal W=\mathcal W_{\omega,\epsilon}\in\mathcal{B}(\mathcal{H}^s)$, for $s\ge0$, such that $t\mapsto u(t,\cdot)\in\mathcal{H}^1$ satisfies 
\begin{equation}\label{mainthmorieqn}
	\mi\partial_tu=(-\Delta+|x|^2)u+\epsilon V(x,\omega t)u
\end{equation}
if and only if $t\mapsto v(t,\cdot)=\Psi(\omega t)^{-1}u(t,\cdot)$ satisfies 
 the autonomous equation
\begin{equation}\label{mainthmeqn}
\mi\partial_t v=(-\Delta+|x|^2)v+\mathcal{W}(v). 
\end{equation} 
 Furthermore, there exists $C>0$ such that
\begin{equation}\label{mainthmestimations}
\begin{alignedat}{4}
&\text{Meas}(\mathcal{D}_0\backslash\mathcal{D}_\epsilon)\le C\epsilon^\frac16,&&\\\
&\|\mathcal{W}\|_{\mathcal{B}(\mathcal{H}^s)}\le C\epsilon,&&\text{ for all }s\ge0,\\
&\|\Psi(\varphi)^{\pm 1}-\Id\|_{\mathcal{B}(\mathcal{H}^p)}\le C\epsilon^{\frac5{12}},~\forall~\varphi\in\mathbb{T}^n_{\sigma/2},&\quad&\text{ for all }p\in[0,1].
\end{alignedat}
\end{equation}
\end{theorem}
\begin{remark}
The infinite matrix $(W_a^b)_{a,b\in \mathcal{E}}$ of the operator $\mathcal{W}$ written in the Hermite basis $(W_a^b=\int_{\mathbb R^d} \Phi_a\mathcal{W}(\Phi_b)dx)$ is block diagonal and satisfies for some $C>0$:
\begin{equation*}
\|W_{[a]}^{[a]}\|\le \frac{C\epsilon}{(1+\ln w_a)^{2\iota}},\quad\forall\,a\in\mathcal{E}.
\end{equation*}
\end{remark}
As a consequence, we obtain the following  corollaries concerning the Sobolev norm estimations on the solution of \eqref{maineqn} and the spectra of the corresponding Floquet operator defined by
\begin{equation*}
\mathrm{K}=-\mi\sum_{j=1}^n\omega_j\frac\partial{\partial\varphi_j}-\Delta+|x|^2+\epsilon V(x,\varphi).
\end{equation*}
\begin{corollary}\label{SobolevNormEstimate}
Assume that $V$ satisfies \eqref{potentialcd} with $\iota\ge\frac{n+d+1}2$. There exists $\epsilon_*>0$ such that for all $0\le\epsilon<\epsilon_*$ and $\omega\in\mathcal{D}_\epsilon$, there exists a unique solution $u\in\mathcal{C}(\mathbb{R},\mathcal{H}^1)$ such that $u(0)=u_0$. Moreover, $u$ is almost-periodic in time and satisfies 
\[
(1-C\epsilon)\|u_0\|_{\mathcal{H}^1}\le\|u(t)\|_{\mathcal{H}^1}\le(1+C\epsilon)\|u_0\|_{\mathcal{H}^1},
\]
for some $C>0$.
\end{corollary}

\begin{corollary}\label{FloquetOperatorSpectrum}
Assume that $V$ satisfies \eqref{potentialcd} with $\iota\ge\frac{n+d+1}{2}$. There exists $\epsilon_*>0$ such that for all $0\le\epsilon<\epsilon_*$ and $\omega\in\mathcal{D}_\epsilon$ the spectrum of the Floquet operator K is a pure point.
\end{corollary}

\subsection{An introduction of the proof and related results}
In this paper we will follow the strategy in \cite{GP2019} or \cite{EK2009}. 
As \cite{GP2019}, 
if we  endow the phase space $\mathcal{H}^s\times\mathcal{H}^s$ with the symplectic structure  $\mi\md u\wedge\md\bar u$, the equation \eqref{mainthmorieqn} can be rewritten as the Hamiltonian system associated with the Hamiltonian
\begin{equation}\label{generalHfunction}
H(u,\bar u)=h(u,\bar u)+\epsilon q(\omega t,u,\bar u),
\end{equation}
where 
$
h(u,\bar u)=\int_{\mathbb{R}^d}|\nabla u|^2+|x|^2u\bar u\md x
$ 
and 
$
q(\omega t,u,\bar u)=\int_{\mathbb{R}^d}V(x,\omega t)u\bar u\md x.
$
Expanding $u$ and $\bar u$ on the Hermite basis of real valued functions
$
u=\sum_{a\in\mathcal{E}}\xi_a\Phi_a,~\bar u=\sum_{a\in\mathcal{E}}\eta_a\Phi_a,
$
the Hamiltonian reads as 
$
h=\sum_{a\in\mathcal{E}}w_a\xi_a\eta_a,~q=\langle\xi,Q(\omega t)\eta\rangle,
$
where $\langle\cdot,\cdot\rangle$ denotes the natural scalar product (no complex conjugation) and $Q$ is the infinite matrix whose entries are
\begin{equation}\label{Qentry}
Q_a^b(\omega t)=\int_{\mathbb{R}^d}V(x,\omega t)\Phi_a(x)\Phi_b(x)\md x,
\end{equation}
and  $(\xi,\eta)\in Y_s$, in which 
$
Y_s=\{\zeta=(\zeta_a\in\mathbb{C}^2,a\in\mathcal{E}):\|\zeta\|_s<\infty\}.
$
Therefore, Theorem $\ref{mainthm}$ is equivalent to the reducibility problem for the non-autonomous Hamiltonian system associated with the Hamiltonian
\begin{equation}\label{generalHamiltonian}
 \sum_{a\in\mathcal{E}}w_a\xi_a\eta_a+\epsilon\langle\xi,Q(\omega t)\eta\rangle.
\end{equation}
As \cite{GP2019} we will  construct a canonical change of variables and  conjugate the Hamiltonian system 
 with the Hamiltonian (\ref{generalHamiltonian})
  to the Hamiltonian system associated with an autonomous Hamiltonian 
$ \sum_{a\in\mathcal{E}}w_a\xi_a\eta_a+\langle\xi,W\eta\rangle,
$
where $W$ is block diagonal and will be clear in the following sections.\par
Here we would like to compare our approaches with those of  Gr\'{e}bert and Paturel \cite{GP2019}. 
By and large, because of the so-called small-divisor problem, in both \cite{GP2019} and the present paper 
the KAM technique is used to eliminate the dependence on  
$\varphi$ of the perturbation $\langle \xi, Q(\varphi)\eta\rangle$.  
There is, however, an apparent difference.  In \cite{GP2019}, the perturbation matrix $Q(\varphi)$ belongs to $ \overline{\mathcal{M}}_{s,\beta}$.  
In other words, for any $a,b\in \mathcal{E}$, 
\begin{eqnarray}\label{grebertdf}
w^{\beta}_aw^{\beta}_b\|Q_{[a]}^{[b]}\|\Big(\frac{\sqrt{\min\{w_a, w_b\}}+|w_a-w_b|}{\sqrt{\min\{w_a, w_b\}}}\Big)^{\frac{s}{2}}\leq C . 
\end{eqnarray}
The estimate comes from Cor.  3.2 from \cite{KT2005}, the assumption of the perturbation and a delicate analysis.  In this paper, we have a weaker assumption on 
$V(x,\varphi)$. Similarly, we also need an estimate on the perturbation matrix element $Q_{[a]}^{[b]}$ for every $a,b\in \mathcal{E}$.  The following is the corresponding estimate:
\begin{lemma}\label{keylemma}
Assume that $V$ satisfies \eqref{potentialcd} with $\iota\ge0$, then we have
\begin{equation}\label{keyestimate}
\left|\int_{\mathbb{R}^d}V(x,\varphi)\Psi_a(x)\Psi_b(x)\md x\right|\le\frac C{(1+\ln w_a)^\iota(1+\ln w_b)^\iota}
\end{equation}
for all $\Psi_a\in E_{[a]},\Psi_b\in E_{[b]}$ and $\|\Psi_a\|=\|\Psi_b\|=1$, where the constant $C\equiv C(d,\iota)>0$ and $\|\cdot\|$ denotes the $L^2(\mathbb{R}^d)$ norm.
\end{lemma}
\begin{remark}
 The above estimate was firstly proved by Z. Wang and one of the authors in \cite{WL2017} when $d=1$. 
 \end{remark}
 \begin{remark}
 From the above lemma, we obtain 
\[
(1+\ln w_a)^\beta(1+\ln w_b)^\beta\|Q_{[a]}^{[b]}\|\leq C. 
\]
Thus we need to introduce  a new space $\mathcal{M}_{\beta}$  which can be compared with the corresponding  one  $\overline{\mathcal{M}}_{s,\beta}$ in \cite{GP2019}.
See section \ref{secreducibility} for the definition of $\mathcal{M}_{\beta}$. 
 \end{remark}
\indent  In  \cite{GP2019} Gr\'ebert and Paturel 
used Lemma 4.3 called the key lemma  to solve the homological equation, in which 
they assume $ |\mu_a-\lambda_a|\le \frac{C_0}{w_a^\delta}$   for all $a\in\E$.  We remark this lemma was first built up in Proposition 2.2.4 \cite{DeSz2004} and also applied in \cite{FGr19} and  \cite{GP2016}.
But in our case we only have a weaker assumption
$ |\mu_a-\lambda_a|\le \frac{C_0}{(1+\ln w_a)^\delta}$ for all $a\in\E$
which comes from $Q\in \mathcal{M}_{\beta}$. 
On  the first try one  can follow  the proof of Lemma 4.3 in \cite{GP2019}, 
the estimate should be   
\begin{equation}\label{oldes}
\|B(k)_{[a]}^{[b]}\|\le  \frac{ C_{d} K^{d-1} \exp\{ C_{\delta,d} \kappa^{- \frac{1}{\delta}}\} }{\kappa(1+|w_a-w_b|)}\|A_{[a]}^{[b]}\|. 
\end{equation}
If we choose $\kappa\sim \epsilon^{\alpha}$, the term $\exp\{C_{\delta,d} \kappa^{- \frac{1}{\delta}}\}$ is too big for KAM iteration whatever $\alpha>0$ is chosen. 
If one  chooses  $\kappa\sim \frac{1}{|\ln \epsilon|^{\delta}}$, it follows $\exp\{ C_{\delta,d} \kappa^{- \frac{1}{\delta}}\}\sim \epsilon^{-\frac{1}{100}} $ which seems to be enough for KAM. 
But from a further investigation we find that one will face big troubles relative with measure estimates.  
Now it is clear that we need to develop some new estimate for $\|B(k)_{[a]}^{[b]}\|$.  
In the following we denote by $\M$ the set of infinite matrices $A:\mathcal{E}\times\mathcal{E}\mapsto\mathbb{C}$ which  satisfies 
$\sup\limits_{a,b\in\mathcal{E}} \|A_{[a]}^{[b]}\|<\infty$,
where $A_{[a]}^{[b]}$ denotes the restriction of $A$ on the block $[a]\times[b]$ and $\|\cdot\|$ denotes the operator norm.   
Here is the new estimate.
\begin{lemma}\label{criticallemma}
	Let $A\in\M$ and $B(k)$ defined for $k\in\mathbb{Z}^n$ with $|k|\le K$ by 
	\begin{equation}\label{clemeqn}
	B(k)_j^l=\frac{A_j^l}{k\cdot\omega-\mu_j+\mu_l},\quad j\in[a],l\in[b],
	\end{equation}
	where $\omega\in[0,2\pi)^n$ and $(\mu_a)_{a\in\E}$ is a sequence of positive real numbers satisfying
	\begin{equation}\label{mudist}
	|\mu_a-\lambda_a|\le \frac{C_0}{(1+\ln w_a)^\delta},  \quad\text{for all }a\in\E
	\end{equation}
 such that for all $a,b\in\E$ and all $|k|\le K$, 
	\begin{align}
	&|k\cdot\omega-\lambda_a+\lambda_b|\ge \gamma(1+|w_a-w_b|),\label{kolambda}\\
	&|k\cdot\omega-\mu_j+\mu_l|\ge\kappa(1+|w_a-w_b|),\quad j\in[a],l\in[b]\label{komu}
	\end{align}
	where 
	 $\gamma, \kappa, \delta, C_0>0$. Then $B(k)\in\M$ and there exists a positive constant $ C_{\delta,d}$  such that
	\begin{equation*}
	\|B(k)_{[a]}^{[b]}\|\le  \frac{ \exp\{ C_{\delta,d}\gamma^{-1/\delta}\}}{\kappa(1+|w_a-w_b|)}\|A_{[a]}^{[b]}\|,
	\end{equation*}
	where $\|\cdot\|$ denotes the operator norm and $C_{\delta,d} = \frac{d-1}{2}(2C_0)^{\frac{1}{\delta}}$.  
\end{lemma}
\begin{remark}
The term $\exp\{ C_{\delta,d}\gamma^{-1/\delta}\}$  is clearly better than the one $\exp\{C_{\delta,d} \kappa^{- \frac{1}{\delta}}\}$
in (\ref{oldes}).  We explain it in a heuristic way. For  example, if choose $\gamma\sim K^{-\alpha}$, it follows that 
$\exp\{ C_{\delta,d}\gamma^{-1/\delta}\} \sim \exp\big\{K^{\frac{\alpha}{\delta}} \big\}$ and $ \exp\big\{K^{\frac{\alpha}{\delta}} \big\}\leq C\epsilon^{-\frac14}$ when $\alpha<\delta$. 
The KAM iteration can be easily set up as \cite{GP2019}. 
\end{remark}
\indent  In the end let us review the previous works on the reducibility of some important PDEs and the behaviors of solutions in Sobolev spaces. 
We begin with  the quantum harmonic oscillators(for short ``QHO''). See \cite{Com87, GT2011, WL2017, Wang2008} for the reducibility results on 1-d QHO with bounded perturbations. Bambusi \cite{Bam2017,Bam2018} firstly built the reducibility results for 1-d QHO with unbounded perturbations.  His  proof was based on the pseudodifferential calculus, in which he  generalized the ideas from \cite{BBM2014, PT2001}(see also \cite{BM2016, FP2015, Mon2014}). It seems that the pseudodifferential calculus method hasn't been applied for the unbounded perturbation terms such as 
$\langle x\rangle^{\mu}\cos (x-\omega t)(\mu>0) $ (Remark 2.7 in \cite{Bam2017}), which have been solved recently by Luo and one of the authors( \cite{LiangLuo2021}) when $0\leq \mu<\frac13$.
For a general 1-d Schr\"odinger equation we recall the classical results \cite{BG2001} in which the potential grows at infinity like $|x|^{2\ell }$ with $\ell >1$ and the perturbation is bounded by $|x|^{\beta}$ with $\beta<\ell-1$. The limit case was solved by Liu-Yuan(\cite{LY2010}). Recently the upper bound for $\beta$ was improved further by Bambusi(\cite{Bam2017, Bam2018}). When the perturbation is limited to a class of terms  such as $\langle x\rangle^{\mu}\cos(x-\omega t)$ we(\cite{LW2020})  can improve the original upper boundedness from $\ell -1$ to at least $\ell -\frac34$.\\ 
\indent Reducibility for PDEs in high dimension was initiated by Eliasson-Kuksin \cite{EK2009}. We can refer to \cite{GP2019} and \cite{LiangW2019} for higher-dimensional QHO with bounded potential. The  reducibility result for n-d QHO was first built  in \cite{BGMR2018} by Bambusi-Gr\'ebert-Maspero-Robert.  Towards other PDEs Montalto \cite{Mon19} obtained the first reducibility result for linear wave equations with unbounded perturbations on $\mathbb T^d$, which can be applied to the linearized Kirchhoff equation in higher dimension.  Bambusi, Langella and Montalto \cite{BLM18} obtained the reducibility results for transportation equations with unbounded perturbations(\cite{FGiMP19}). See also  \cite{FGr19, FGN19} for a linear Schr\"odinger equation on zoll manifold with unbounded potential.  We remark that by implementing the above techniques  the KAM-type results of  quasi - linear PDEs such as  incompressible Euler flows in 3D  \cite{BaMon21}  and forced Kirchhoff equation on $\mathbb T^d$ \cite{CoMon18} have been built recently. \\
 \indent The reducibility results usually imply the boundedness of Sobolev norms. For the growth rate of the solutions with time in  $\mathcal{H}^s-$norm,  see \cite{BGMR2018} and  \cite{GrYa2000} for a $t^{s}$- polynomial growth for 1-d QHO with time periodic perturbations.  Delort \cite{Del2014} constructed a $t^{s/2}$- polynomial growth for 1-d QHO with certain time periodic order zero perturbation(see \cite{Mas2018} for a short proof).  Combining with the ideas in \cite{BGMR2018} and \cite{Eli1992}, Zhao, Zhou and one of the authors \cite{LZZ2020} obtained the precise dynamics for one class of 1-d QHO  with quasi - periodic in time quadratic perturbations, in which they also presented an exponential growth in time  for 1-d QHO.  Recently,  for 2-d QHO with perturbation which is decaying in $t$, Faou-Rapha\"el \cite{FaRa2020} constructed a solution whose $\mathcal{H}^1-$norm presents logarithmic growth with t. For 2-d QHO with perturbation being the projection onto Bargmann-Fock  space, Thomann \cite{Th2020}  constructed explicitly  a traveling wave whose Sobolev norm presents polynomial growth with $t$, based on the study in \cite{ScTh2020} for linear Lowest Landau equations(LLL) with a time-dependent potential.  There are also many literatures, e.g. \cite{BLM2021, BGMR2019, BM2019, Bou99a, Bou99b, FZ12, MR2017, WWM2008}, which are closely relative to the upper growth bound of the solution in Sobolev space. \\
 \indent The rest paper will be organized as follows. In Section 2,  a new reducibility theorem is presented. In Section 3, we check all the hypothesis of the  reducibility theorem  are satisfied which follows  the main theorem. In Section 4, after present the proof of key lemma  we prove the  reducibility theorem.  Finally, the appendix contains some technical Lemmas. 
 
\section{Reducibility Theorem}\label{secreducibility}
In this section we state an abstract reducibility theorem for quadratic quasiperiodic in time Hamiltonians of a general form
$
\sum_{a\in\mathcal{E}}\lambda_a\xi_a\eta_a+\epsilon\langle\xi,Q(\omega t)\eta\rangle.
$
\subsection{Setting}
Following \cite{GP2019}, we will introduce some spaces and relative algebraic properties for later use.

\noindent\textbf{Linear space.} Let $s\in\mathbb{R}$, we consider the complex weighted $\ell^2$ space
\[
\ell_s^2=\{\xi=(\xi_a\in\mathbb{C},a\in\mathcal{E}):\|\xi\|_s<\infty\},\text{ where }\|\xi\|_s^2=\sum_{a\in\mathcal{E}}w_a^s|\xi_a|^2.
\]
Then we define 
\[
Y_s=\ell_s^2\times\ell_s^2=\{\zeta=(\zeta_a=(\xi_a,\eta_a)\in\mathbb{C}^2,a\in\mathcal{E}):\|\zeta\|_s<\infty\},
\]
where 
\[
\|\zeta\|_s^2=\sum_{a\in\mathcal{E}}w_a^s|\zeta_a|^2=\sum_{a\in\mathcal{E}}w_a^s(|\xi_a|^2+|\eta_a|^2).
\]
We equip the space $Y_s,~s\ge0$, with the symplectic structure $\mi \md\xi\wedge \md\eta$. Let $f$ be a $\mathcal{C}^1$ smooth function, defined on a domain $\mathcal O\subset Y_s$, then we have the associated Hamiltonian system
\[
\begin{cases}
\dot{\xi}=-\mi\nabla_\eta f(\xi,\eta),&\\
\dot{\eta}=\mi\nabla_\xi f(\xi,\eta),&
\end{cases}
\] 
where $\nabla f=(\nabla_\xi f,\nabla_\eta f)^T$ is the gradient with respect to the scalar product in $Y_0$. Naturally, define the Poisson bracket for any $\mathcal{C}^1$ smooth functions $f$ and $g$, defined on a domain $\mathcal{O}\subset Y_s$
\[
\{f,g\}=-\mi\big(\sum_{a\in\mathcal{E}}\frac{\partial f}{\partial\xi_a}\cdot\frac{\partial g}{\partial\eta_a}-\frac{\partial g}{\partial\xi_a}\cdot\frac{\partial f}{\partial\eta_a}\big).
\]

\noindent\textbf{Infinite matrices.} Let $\beta\ge0$, we denote by $\M_\beta$ the set of infinite matrices $A:\mathcal{E}\times\mathcal{E}\mapsto\mathbb{C}$ that satisfy
\begin{equation*}
|A|_\beta:=\sup_{a,b\in\mathcal{E}}(1+\ln w_a)^\beta(1+\ln w_b)^\beta\|A_{[a]}^{[b]}\|<\infty,
\end{equation*}
where $A_{[a]}^{[b]}$ denotes the restriction of $A$ on the block $[a]\times[b]$ and $\|\cdot\|$ denotes the operator norm. Further we denote $\M=\M_0$. We will also need the more regular space $\M_\beta^+\subset\M_\beta$: an infinite matrix $A\in\M$ belongs to $\M_\beta^+$ if
\begin{equation*}
|A|_{\beta+}:=\sup_{a,b\in\mathcal{E}}(1+|w_a-w_b|)(1+\ln w_a)^\beta(1+\ln w_b)^\beta\|A_{[a]}^{[b]}\|<\infty.
\end{equation*} 

The following structural lemma is proved in Appendix \ref{proofstructurelemma}:
\begin{lemma}\label{structure}
Let $\beta>\frac12$, there exists an absolute constant $C\equiv C(\beta)>0$ such that
\begin{enumerate}
	\item[(i).\phantom{ii}] Let $A\in\M_\beta$ and $B\in\M_\beta^+$. Then $AB$ and $BA$ belong to $\M_\beta$ and 
\[
|AB|_\beta,|BA|_\beta\le C|A|_{\beta}|B|_{\beta+}.
\]
\item[(ii).\phantom{i}] Let $A,B\in\M_\beta^+$. Then $AB$ belongs to $\M_\beta^+$ and 
$
|AB|_{\beta+}\le C|A|_{\beta+}|B|_{\beta+}.
$
\item[(iii).] Let $A\in\M_\beta^+$. Then $\me^A-\Id$ belongs to $\M_\beta^+$ and
$
|\me^A-\Id|_{\beta+}\le\me^{C|A|_{\beta+}}|A|_{\beta+}.
$
\item[(iv).] Let $A\in\M_\beta$. Then for any $s\ge1$, $A\in\mathcal{B}(\ell_s^2,\ell_{-s}^2)$ and 
$
\|A\|_{\mathcal{B}(\ell_s^2,\ell_{-s}^2)}\le C|A|_\beta.
$
\item[(v).\phantom{i}] Let $A\in\M_\beta^+$. Then for any $s\in[-1,1]$,  $A\in\mathcal{B}(\ell_s^2)$ and
$
\|A\|_{\mathcal{B}(\ell_s^2)}\le C|A|_{\beta+}.
$
\end{enumerate}
\end{lemma}

\noindent\textbf{Normal form.}
\begin{definition}\label{normalform}
A matrix $Q:\mathcal{E}\times\mathcal{E}\mapsto\mathbb{C}$ is in normal form, denoted by $Q\in\nf$, if
\begin{enumerate}
	\item[(i).\phantom{i}] $Q$ is block diagonal, i.e. $Q_a^b=0$ for all $w_a\ne w_b$.
	\item[(ii).] $Q$ is Hermitian, i.e. $Q_b^a=\overline{Q_{a}^{b}}$.
\end{enumerate}
\end{definition}
Notice that a block diagonal matrix with bounded blocks in operator norm defines a bounded operator on $\ell_s^2,~\forall s\in \mathbb R$, which implies $\M_\beta\cap\nf\subset\mathcal{B}(\ell_s^2)$. To a matrix $Q=(Q_a^b)_{a,b\in\mathcal{E}}\in\mathcal{B}(\ell_s^2,\ell_{-s}^2)$ for $s\ge0$, define in a unique way a quadratic form on $Y_s\ni(\zeta_a)_{a\in\mathcal{E}}=(\xi_a,\eta_a)_{a\in\mathcal{E}}$ by the  formula
$
q(\xi,\eta)=\langle\xi,Q\eta\rangle=\sum_{a,b\in\mathcal{E}}Q_a^b\xi_a\eta_b.
$
A straightforward computation leads to
\begin{equation}\label{liebracket}
\{q_1^{},q_2^{}\}(\xi,\eta)=\mi\langle\xi,[Q_1,Q_2]\eta\rangle,
\end{equation}
where $[Q_1,Q_2]$ is the communicator of two matrices $Q_1$ and $Q_2$.

In the following, by abuse of language, we will call both of $q$ and $Q$ the Hamiltonian.

\noindent\textbf{Parameter.} In all the paper $\omega$ will play the role of a parameter belonging to $\mathcal{D}_0=[0,2\pi)^n$. All the constructed maps will be dependent on $\omega$ with $\mathcal{C}^1$ regularity. When a map is only defined on a Cantor subset of $\mathcal{D}_0$ the regularity has to be understood in the Whitney sense.

\noindent\textbf{A class of quadratic Hamiltonian.} Let $\beta>\frac12,\,\sigma>0$ and $\mathcal{D}\subset\mathcal{D}_0$ and denote by $\M_\beta(\mathcal{D},\sigma)$ the set of $\mathcal{C}^1$ mappings
$
\mathcal{D}\times\mathbb{T}_\sigma^n\ni(\omega,\varphi)\mapsto Q(\omega,\varphi)\in\M_\beta,
$
which is real analytic in $\varphi\in\mathbb{T}_\sigma^n$.  Endow the space with the norm 
\begin{equation*}
[Q]_\beta^{\mathcal{D},\sigma}:=\sup_{\substack{\omega\in\mathcal{D},j=0,1\\ |\Im\varphi|<\sigma}}|\partial_\omega^jQ(\omega,\varphi)|_\beta.
\end{equation*} 
In view of Lemma \ref{structure} (iv), to a matrix $Q\in\M_\beta(\mathcal{D},\sigma)$ define  the quadratic form in $Y_1$
$
q(\xi,\eta;\omega,\varphi)=\langle\xi,Q(\omega,\varphi)\eta\rangle
$
and a straightforward computation leads to
\[
|q(\xi,\eta;\omega,\varphi)|\le C[Q]_\beta^{\mathcal{D},\sigma}\|(\xi,\eta)\|_1^2,\quad\text{for}~(\xi,\eta)\in Y_1,~(\omega,\varphi)\in\mathcal{D}\times\mathbb{T}_\sigma^n.
\]
 Denote by $\M_\beta^+(\mathcal{D},\sigma)$ the subspace of $\M_\beta(\mathcal{D},\sigma)$ formed by Hamiltonians $S$, verifying $S(\omega,\varphi)\in\M_\beta^+$  and endow it with the norm
$
[S]_{\beta+}^{\mathcal{D},\sigma}:=\sup_{\substack{\omega\in\mathcal{D},j=0,1\\ |\Im\varphi|<\sigma}}|\partial_\omega^jS(\omega,\varphi)|_{\beta+}.
$
The space of Hamiltonians $N\in\M_\beta(\mathcal{D},\sigma)$ that are independent on $\varphi$ will be denoted by $\M_\beta(\mathcal{D})$ and equipped with the norm
$
[N]_\beta^\mathcal{D}:=\sup_{\omega\in\mathcal{D},j=0,1}|\partial_{\omega}^jN(\omega)|_\beta.
$

\noindent\textbf{Hamiltonian flow.} To any $R\in\M_\beta^+$ with $\beta>\frac12$, define in a unique way the symplectic linear change of variables on $Y_s$:
$
(\xi,\eta)\mapsto(\me^{-\mi R^T}\negthickspace\xi,\me^{\mi R}\eta).
$
It is well-defined and invertible in $\mathcal{B}(\ell_s^2)$ for all $s\in[-1,1]$ by the assertion (v) of Lemma \ref{structure}. Concretely, the change of variables can be regarded as the time-one flow generated by the quadratic Hamiltonian $\chi(\xi,\eta)=\langle\xi,R\eta\rangle$ and it preserves the symmetry $\eta=\bar{\xi}$ for any initial datum considered in the paper if and only if $R$ is a Hermitian matrix, i.e.
\begin{equation}\label{hermitian}
R^T=\overline{R}.
\end{equation}
When $R$ also depends smoothly on $\varphi,~\mathbb{T}_\sigma^n\ni\varphi\mapsto R(\varphi)\in\M_\beta^+$ we associate to $R$ the symplectic linear change of variables on the extended phase space
\begin{equation}\label{change}
\phi_\chi^1(y,\varphi,\xi,\eta)=(\tilde{y},\varphi,\me^{-\mi R^T}\negthickspace\xi,\me^{\mi R}\eta),
\end{equation}
where  the Hamiltonian $\chi(y,\varphi,\xi,\eta)=\langle\xi,R(\varphi)\eta\rangle$ and
$
\tilde{y}=y-\int_0^1\langle\me^{-\mi tR^T}\negthickspace\xi,\nabla_\varphi R\me^{\mi tR}\eta\rangle\md t.
$
\subsection{Hypothesis on the spectrum}
Now we present our hypothesis on $\lambda_a,\,a\in\mathcal{E}$:
\begin{hypo1}[asymptotics] Assume that there exists an absolute constant $c_0>0$ such that
\begin{equation}\label{asymptotics}
\lambda_a\ge c_0\text{ and }|\lambda_a-\lambda_b|\ge c_0|w_a-w_b|,\text{ for all } a,b\in\mathcal{E}.
\end{equation}
\end{hypo1}
\begin{hypo2}[second Melnikov condition in measure] There exist absolute positive  constants $\tau_1,\tau_2$ and $C$ such that the following holds: for each $\gamma>0$ and $K\ge1$ there exists a closed subset $\mathcal{D}'=\mathcal{D}'(\gamma,K)\subset\mathcal{D}_0$ satisfying
\begin{equation*}
\text{Meas}(\mathcal{D}_0\setminus\mathcal{D}')\le CK^{\tau_1}\gamma^{\tau_2}
\end{equation*}
such that for all $\omega\in\mathcal{D}'$, all $k\in\mathbb{Z}^n$ with $0<|k|\le K$ and all $a,b\in\mathcal{E}$ we have
\begin{equation*}
|k\cdot\omega-\lambda_a+\lambda_b|\ge\gamma(1+|w_a-w_b|).
\end{equation*}
\end{hypo2}

\subsection{The reducibility theorem}
Consider the Hamiltonian
\begin{equation}\label{oriHamiltonian}
H_\omega(t,\xi,\eta)=\sum_{a\in\mathcal{E}}\lambda_a\xi_a\eta_a+\epsilon\langle\xi,Q(\omega t)\eta\rangle
\end{equation}
and the associated non-autonomous Hamiltonian system on $Y_s$:
\begin{equation}\label{oriHamiltoniansys}
\begin{cases}
\dot{\xi}=-\mi N_0\xi-\mi\epsilon Q^T(\omega t)\xi,&\\
\dot{\eta}=\mi N_0\eta+\mi\epsilon Q(\omega t)\eta,&
\end{cases}
\end{equation}
where $N_0=\text{diag}(\lambda_a:a\in\mathcal{E})$.

\begin{theorem}\label{reduthm}
Let $\alpha=\frac{\alpha_1}{\alpha_2}+1$ with $\alpha_1=\max\{\tau_1,n+d\}$ and $\alpha_2=\max\{\tau_2,1\}$. Fix $\beta\ge\frac\alpha2$ and $\sigma>0$. Assume that $(\lambda_a)_{a\in\mathcal{E}}$ satisfies Hypothesis H1, H2 and that $Q\in\M_\beta(\mathcal{D}_0,\sigma)$, then there exists $\epsilon_*>0$ such that for all $0\le\epsilon<\epsilon_*$, there exists $C>0$ and
\begin{enumerate}
	\item[(i).\phantom{ii}] a Cantor set $\mathcal{D}_\epsilon\subset\mathcal{D}_0$ of asymptotically full measure:
\begin{equation}\label{DcloseD0}
	\text{Meas}(\mathcal{D}_0\setminus\mathcal{D}_\epsilon)\le C\epsilon^{\alpha_2/6};
\end{equation}
\item[(ii).\phantom{i}] a $\mathcal{C}^1$ family (in $\omega\in\mathcal{D}_\epsilon$) of real analytic (in $\varphi\in\mathbb{T}_{\sigma/2}^n$) linear, unitary and symplectic coordinate transformations on $Y_0$:
\[
\Phi_{\omega}(\varphi):(\xi,\eta)\mapsto(\overline{M_\omega(\varphi)}\xi,M_\omega(\varphi)\eta),\quad(\omega,\varphi)\in\mathcal{D}_\epsilon\times\mathbb{T}_{\sigma/2}^n;
\]
\item[(iii).] a $\mathcal{C}^1$ family of quadratic autonomous Hamiltonians in normal form 
$\mathcal{H}_\omega=\langle\xi,N_\omega\eta\rangle$, 
where $N_\omega\in\nf$ is close to $N_0$:
\begin{equation}\label{NcloseN0}
\|N_\omega-N_0\|_\beta\le C\epsilon,\quad\omega\in\mathcal{D}_\epsilon,
\end{equation}
such that $t\mapsto(\xi(t),\eta(t))\in Y_1$ is a solution of autonomous Hamiltonian system associated with $\mathcal{H}_\omega$:
\[
\begin{cases}
\dot{\xi}=-\mi N^{T}_\omega\xi&\\
\dot{\eta}=\mi N_\omega\eta
\end{cases}
\]
if and only if $t\mapsto\Phi_\omega(\omega t)(\xi(t),\eta(t))\in Y_1$ is a solution of the original Hamiltonian system \eqref{oriHamiltoniansys}. Furthermore,  $\Phi_\omega(\varphi)$ is  a bounded operator from $Y_p$ into itself for all $p\in[0,1]$ and close to identity:
\begin{equation}\label{Mcloseidentity}
\|M_\omega(\varphi)-\Id\|_{\mathcal{B}(\ell_p^2)}\le C\epsilon^\frac5{12}.
\end{equation} 
\end{enumerate}
\end{theorem}

\section{Applications to the quantum harmonic oscillator on $\mathbb{R}^d$}
In this section we will prove Theorem \ref{reduthm} which concludes Theorem \ref{mainthm}.
\subsection{Verification of the hypothesis}
We first verify the Hypothesis of Theorem \ref{reduthm} for the quantum harmonic oscillator equation \eqref{mainthmorieqn}.
\begin{lemma}[see \cite{GP2016}]\label{verificationHypothesis}
	When $\lambda_a=w_a,~a\in\mathcal{E}$. Hypothesis H1 and H2 hold true with $c_0=1$ and $\mathcal{D}_0=[0,2\pi)^n$ and $\tau_1=n+1,\tau_2=1$.
\end{lemma}
\noindent Proof of Lemma \ref{keylemma}:
	Let $\Psi_a\in E_{[a]},\Psi_b\in E_{[b]}$ and $\|\Psi_a\|=\|\Psi_b\|=1$, then
	\begin{equation}\label{intab}
	\lt|\int_{\mathbb{R}^d}V(x,\varphi)\Psi_a(x)\Psi_b(x)\md x\rt|\le\|V^\frac12(x,\varphi)\Psi_a(x)\|\|V^\frac12(x,\varphi)\Psi_b(x)\|.
	\end{equation}
   By the symmetry we only need to estimate $\|V^\frac12(x,\varphi)\Psi_a(x)\|$.
   
   Let $\mathbb{D}_a=\{x\in\mathbb{R}^d:~|x|\le w_a^{1/(4d)}\}$ and denote by $\mathbb{D}_a^c$ its complement, then one has 
	\begin{equation}\label{inta}
\|V^\frac12(x,\varphi)\Psi_a(x)\|\le\|V^{\frac12}(x,\varphi)\Psi_a(x)\|_{L^2(\mathbb{D}_a)}+\|V^{\frac12}(x,\varphi)\Psi_a(x)\|_{L^2(\mathbb{D}_a^c)}.
	\end{equation} 
 Clearly, since $|V(x,\varphi)|\le C$ then one has 
	\begin{alignat}{4}
	&&\hspace{1pt}&\|V^{\frac12}(x,\varphi)\Psi_a(x)\|_{L^2(\mathbb{D}_a)}\le\|\Psi_a(x)\|_{L^p(\mathbb{D}_a)}\|V^{\frac12}(x,\varphi)\|_{L^q(\mathbb{D}_a)}\qquad\underline{~\frac1p+\frac1q=\frac12~}\notag\\
	&\le&&C\|\Psi_a(x)\|_{L^p(\mathbb{R}^d)}\lt(\text{Meas}(\mathbb{D}_a)\rt)^{\frac1q}\notag\\
	&\le&&Cw_a^{-\frac12(\frac12-\frac1p)}w_a^{1/(4q)}\hspace{1cm}\underline{{\text{ by Cor. 3.2 in  \cite{KT2005} with }p=\frac{2(d+2)}{d+1}}~}\notag\\
	&\le&&Cw_a^{-\frac1{4q}}=Cw_a^{-\frac1{8(d+2)}}.\label{inta1}
	\end{alignat}
	Besides, for $x\in\mathbb{D}_a^c$ we have $|x|\ge w_a^{1/(4d)}$. Then 
	$
	1+\ln(1+|x|^2)\ge 1+\ln w_a^{1/(2d)}\ge\frac{1+\ln w_a}{2d}.
	$
	By \eqref{potentialcd} one has 
	$
	|V(x,\varphi)|\le\frac{C}{(1+\ln(1+|x|^2))^{2\iota}}\le\frac{C_{d,\iota}}{(1+\ln w_a)^{2\iota}}
	\quad\text{for }x\in\mathbb{D}_a^c.
	$
	Hence, we obtain
	\[
	\|V^{\frac12}(x,\varphi)\Psi_a(x)\|_{L^2(\mathbb{D}_a^c)}\le\frac{C_{d,\iota}}{(1+\ln w_a)^{\iota}}\|\Psi_a(x)\|_{L^2(\mathbb{R}^d)}\le\frac{C_{d,\iota}}{(1+\ln w_a)^{\iota}}.
	\]
	Combining \eqref{inta}, \eqref{inta1} with the last estimation we  have
	\begin{equation}\label{inta2}
\|V^\frac12(x,\varphi)\Psi_a(x)\|\le\frac{C_{d,\iota}}{(1+\ln w_a)^\iota}.
	\end{equation}
	Collecting \eqref{intab} and \eqref{inta2} concludes the result \eqref{keyestimate}.\qed
	
\begin{remark}\label{3.4}
Using Lemma \ref{keylemma} and a similar method as \cite{GP2016} we have $Q$, 
defined by 
$
Q_a^b(\varphi)=\int_{\mathbb{R}^d}V(x,\varphi)\Phi_a(x)\Phi_b(x)\md x,
$
belongs to 
$\M_\iota(\mathcal{D}_0,\sigma)$.
\end{remark}

\begin{remark}
For quantum harmonic oscillator equation \eqref{mainthmorieqn}, we have  $\alpha=n+d+1$  since $\alpha=\frac{\alpha_1}{\alpha_2}+1$ with $\alpha_1=\max\{\tau_1,n+d\}$ and $\alpha_2=\max\{\tau_2,1\}$, where 
$\tau_1=n+1$ and $\tau_2=1$. 
\end{remark}

\subsection{Proof of Theorem \ref{mainthm}}
The Schr\"odinger equation \eqref{mainthmorieqn} is a Hamiltonian system on $\mathcal{H}^s\times\mathcal{H}^s~(s\ge1)$ associated with the Hamiltonian function \eqref{generalHfunction}. Written on the orthonormal basis $(\Phi_a)_{a\in\mathcal{E}}$, it is equivalent to the Hamiltonian system 
 on $Y_s$ associated with \eqref{generalHamiltonian} which reads as \eqref{oriHamiltonian} with $\lambda_a=w_a$ and $Q$ given by \eqref{Qentry}. By Lemma \ref{verificationHypothesis} and Remark \ref{3.4}, if $V$ satisfies \eqref{potentialcd} with $\iota\ge\frac{n+d+1}2$, we can apply Theorem \ref{reduthm} to the Hamiltonian\eqref{generalHamiltonian} which concludes Theorem \ref{maineqn}. More precisely, in the new coordinates under a unitary transformation given by Theorem \ref{reduthm},
\(
(\xi,\eta)=(\overline{M_\omega(\omega t)}\xi',M_\omega(\omega t)\eta'),
\) the original system 
\begin{equation*}
\begin{cases}
\dot\xi_a=-\mi w_a\xi_a-\mi\epsilon(Q^T(\omega t)\xi)_a,&\\
\dot\eta_a=\mi w_a\eta_a+\mi\epsilon(Q(\omega t)\eta)_a,&
\end{cases}a\in\mathcal{E}
\end{equation*}
conjugates to an autonomous system as follows:
\[
\begin{cases}
\dot\xi'_a=-\mi(N_\omega^{T}\xi')_a,&\\
\dot\eta'_a=\mi(N_\omega\eta')_a,&
\end{cases}a\in\mathcal{E},
\]
where $N_\omega\in\nf$ and $\overline{M_\omega(\omega t)}M_\omega^T(\omega t)=\Id$. Furthermore, corresponding to the initial datum $u_0(x)=\sum_{a\in\mathcal{E}}\xi(0)_a\Phi_a(x)\in\mathcal{H}^1$  the solution $u(t,x)$ of \eqref{mainthmorieqn} reads
\[
u(t,x)=\sum_{a\in\mathcal{E}}\xi(t)_a\Phi_a(x)\text{ with }
\xi(t)=\overline{M_\omega(\omega t)}e^{-\mi tN_\omega^T}M_\omega^T(0)\xi(0).
\]
Concretely, define the transformation $\Psi(\varphi)\in\mathcal{B}(\mathcal{H}^s)$ by 
\[
\Psi(\varphi)\left(\sum_{a\in\mathcal{E}}\xi'_a\Phi_a(x)\right)
=\sum_{a\in\mathcal{E}}\left(\overline{M_\omega(\varphi)}\xi'\right)_a\Phi_a(x).
\]
Then $v(t,x)$ satisfies \eqref{mainthmeqn}  if and only if $u(t,x)=\Psi(\omega t)v(t,x)$ satisfies the original equation \eqref{mainthmorieqn}, where $\mathcal W$ is defined by
\[
\mathcal W\big(\sum_{a\in\mathcal{E}}\xi'_a\Phi_a(x)\big)=\sum_{a\in\mathcal{E}}(W\xi')_a\Phi_a(x)\text{ with }W=N_\omega-N_0.
\] 
By construction collecting \eqref{DcloseD0}-\eqref{Mcloseidentity} leads to \eqref{mainthmestimations} in Theorem \ref{mainthm}.\qed
\\ For the proofs of Corollary \ref{SobolevNormEstimate} and \ref{FloquetOperatorSpectrum}  refer to \cite{GP2016}.

\section{Proof of Reducibility Theorem}\label{secKAMtools}
In this section we will prove the reducibility theorem presented in Sec.\ref{secreducibility} by the KAM methods.   As we mentioned before, Lemma \ref{criticallemma} 
is very important for solving  homological equations, whose proof we present below.
\subsection{Proof of Lemma \ref{criticallemma}}
\begin{proof}
Denote by $D_{[a]}$ the diagonal (square) matrix with entries $\mu_j$ for $j\in[a]$. The equation \eqref{clemeqn} reads
	\begin{equation}\label{clemeqnl}
	k\cdot\omega B_{[a]}^{[b]}-D_{[a]}B_{[a]}^{[b]}+B_{[a]}^{[b]}D_{[b]}=A_{[a]}^{[b]}.
	\end{equation}
	From \eqref{kolambda} for all $a,b\in\E$ and all $|k|\le K$, $k\cdot\omega-\lambda_a+\lambda_b\neq 0$. We distinguish two cases. \\
	\textbf{Case 1:} suppose that $k\cdot\omega-\lambda_a+\lambda_b>0$. Clearly, in this case by \eqref{kolambda} we have
	$
	k\cdot\omega-\lambda_a+\lambda_b\ge \gamma(1+|w_a-w_b|)\ge \gamma.
	$
	By \eqref{mudist} for $j\in[a]$ and $l\in[b]$, if $ \min\{w_a,w_b\}>\exp\lt\{\lt(\frac{2C_0}\gamma\rt)^{1/\delta}\rt\}$,  we obtain
	\[
	\begin{split}
	k\cdot\omega-\mu_j+\mu_l&=k\cdot\omega-\lambda_a+\lambda_b-\mu_j+\lambda_a-\lambda_b+\mu_l\\
	&\ge \gamma-\frac{C_0}{(1+\ln w_a)^\delta}-\frac{C_0}{(1+\ln w_b)^\delta}\\
	&\ge \gamma-\frac{2C_0}{(1+\ln \min\{w_a,w_b\})^\delta}>0.
	\end{split}
	\]
	It follows that 
	\begin{equation}\label{minmaxl}
	k\cdot\omega+\min_{l\in[b]}\{\mu_l\}>\max_{j\in[a]}\{\mu_j\}>0.
	\end{equation} 
        Subcase 1:  $\min\{w_a,w_b\}>\exp\lt\{\lt(\frac{2C_0}\gamma\rt)^{1/\delta}\rt\}$. 
        (\ref{minmaxl}) proves that $k\cdot\omega I_{[b]}+D_{[b]}$ is an invertible operator and 
	$
	\|\lt(k\cdot\omega I_{[b]}+D_{[b]}\rt)^{-1}\|=\frac1{k\cdot\omega+\min_{l\in[b]}\{\mu_l\}}.$
	Thus \eqref{clemeqnl} is equivalent to 
	\[
	B_{[a]}^{[b]}-D_{[a]}B_{[a]}^{[b]}\lt(k\cdot\omega I_{[b]}+D_{[b]}\rt)^{-1}=A_{[a]}^{[b]}\lt(k\cdot\omega I_{[b]}+D_{[b]}\rt)^{-1}.
	\]
Denote by $\mathscr L_{[a]\times[b]}$ the operator acting on matrices of size $[a]\times[b]$ such that
	\[
	\mathscr L_{[a]\times[b]}(B_{[a]}^{[b]})=D_{[a]}B_{[a]}^{[b]}\lt(k\cdot\omega I_{[b]}+D_{[b]}\rt)^{-1}.
	\]
	Then we have
	\(\displaystyle
	\|\mathscr L_{[a]\times[b]}(B_{[a]}^{[b]})\|\le\frac{\max_{j\in[a]}\{\mu_j\}}{k\cdot\omega+\min_{l\in[b]}\{\mu_l\}}\|B_{[a]}^{[b]}\|
	\) and $\|\mathscr L_{[a]\times[b]}\|<1$ by \eqref{minmaxl}. The operator $\Id-\mathscr L_{[a]\times[b]}$ is invertible and thus 
	\begin{eqnarray}
	\|B(k)_{[a]}^{[b]}\|&=&\|\lt(\Id-\mathscr L_{[a]\times[b]}\rt)^{-1}A_{[a]}^{[b]}\lt(k\cdot\omega I_{[b]}+D_{[b]}\rt)^{-1}\|   \nonumber\\ 
	&\le &\frac1{\displaystyle 1-\frac{\max_{j\in[a]}\{\mu_j\}}{k\cdot\omega+\min_{l\in[b]}\{\mu_l\}}}\cdot\frac{\|A_{[a]}^{[b]}\|}{k\cdot\omega+\min_{l\in[b]}\{\mu_l\}} \nonumber\\
	&=&\frac{\|A_{[a]}^{[b]}\|}{k\cdot\omega+\min_{l\in[b]}\{\mu_l\}-\max_{j\in[a]}\{\mu_j\}}            \label{cancallation}       \\
	&\le &\frac{\|A_{[a]}^{[b]}\|}{\kappa(1+|w_a-w_b|)}\qquad\underline{~\text{by \eqref{komu} and \eqref{minmaxl}}~}. \nonumber
	\end{eqnarray}
	Subcase 2:  $\min\{w_a,w_b\} \le\exp\lt\{\lt(\frac{2C_0}\gamma\rt)^{1/\delta}\rt\}$. 
	In this situation we have $|B(k)_j^l|\le\frac{|A_j^l|}{\kappa(1+|w_a-w_b|)}$. Then for any $\xi\in\ell_0^2$
\[	\|B(k)_{[a]}^{[b]}\xi_{[b]}\|^2=\sum_{j\in[a]}\Big|\sum_{l\in[b]}B(k)_j^l\xi_l\Big|^2\le
	\frac1{\kappa^2(1+|w_a-w_b|)^2}\sum_{j\in[a]}\lt(\sum_{l\in[b]}|A_j^l||\xi_l|\rt)^2.\]
	On the other hand, 
         \[
	\begin{split}
	\sum_{j\in[a]}\lt(\sum_{l\in[b]}|A_j^l||\xi_l|\rt)^2&\le\sum_{j\in[a]}\lt(\sum_{l\in[b]}|A_j^l|^2\rt)\lt(\sum_{l\in[b]}|\xi_l|^2\rt)\le\|\xi_{[b]}\|^2\sum_{j\in[a]}\sum_{l\in[b]}|A_j^l|^2\\
	&\le\|\xi_{[b]}\|^2\sum_{j\in[a]}\|A_{[a]}^{[b]}\|^2\le w_a^{d-1}\|A_{[a]}^{[b]}\|^2\|\xi_{[b]}\|^2.
	\end{split}
	\]
	Similarly, 
         \[
	\begin{split}
	\sum_{j\in[a]}\lt(\sum_{l\in[b]}|A_j^l||\xi_l|\rt)^2&\le\|\xi_{[b]}\|^2\sum_{j\in[a]}\sum_{l\in[b]}|A_j^l|^2\le\|\xi_{[b]}\|^2\sum_{ l\in[b]}\sum_{ j\in[a]}|A_j^l|^2 \\
	&\le\|\xi_{[b]}\|^2\sum_{  l\in[b]}\|A_{[a]}^{[b]}\|^2\le w_b^{d-1}\|A_{[a]}^{[b]}\|^2\|\xi_{[b]}\|^2.
	\end{split}
	\]
	It follows that 
	$$\sum_{j\in[a]}\lt(\sum_{l\in[b]}|A_j^l||\xi_l|\rt)^2\leq (\min\{w_a,w_b\})^{d-1} \|A_{[a]}^{[b]}\|^2\|\xi_{[b]}\|^2. $$
	Therefore,
        \begin{equation*}
	\|B(k)_{[a]}^{[b]}\|\le\frac{ \exp\{ C_{\delta,d}\gamma^{-1/\delta}\}}{\kappa(1+|w_a-w_b|)}\|A_{[a]}^{[b]}\|,
	\end{equation*}
	where $C_{\delta,d}=\frac{(d-1)}{2}(2C_0)^{1/\delta}$.\\
	\textbf{Case 2:} suppose that $k\cdot\omega-\lambda_a+\lambda_b<0$. Similarly, by \eqref{kolambda} we have $\lambda_a-\lambda_b-k\cdot\omega\ge \gamma(1+|w_a-w_b|)\ge \gamma$. Then \eqref{clemeqnl} is equivalent to
	\begin{equation*}
	-\lt(D_{[a]}-k\cdot\omega I_{[a]}\rt)B_{[a]}^{[b]}+B_{[a]}^{[b]}D_{[b]}=A_{[a]}^{[b]}.
	\end{equation*}
The following proof is similar as case 1 in which we use  $\mathscr L'_{[a]\times[b]}$ instead of $\mathscr L_{[a]\times[b]}$,  
where 
	\[
	\mathscr L'_{[a]\times[b]}(B_{[a]}^{[b]})=\lt(D_{[a]}-k\cdot\omega I_{[a]}\rt)^{-1}B_{[a]}^{[b]}D_{[b]}.
	\]
\end{proof}

\begin{remark}
In (\ref{cancallation}), we use a cancellation. 
\end{remark}

\subsection{Homological equation}
In this section consider a homological equation of the form
\begin{equation*}
-\omega\cdot\nabla_{\varphi}S+\mi[N,S]+Q=\text{remainder}
\end{equation*}
where $N\in\nf$ close to $N_0$ and $Q\in\M_\beta$. We will construct a solution $S\in\M_\beta^+$ in the following proposition.
\begin{proposition}\label{iterationproposition}
Denote $\mathcal{D}\subset\mathcal{D}_0$. Let $\mathcal{D}\ni\omega\mapsto N(\omega)\in\nf$ be a $\mathcal{C}^1$ mapping that verifies
\begin{equation}\label{NcloseN0prop}
[N-N_0]_\beta^{\mathcal{D}}\le2\epsilon_0 
\end{equation}
and $Q\in\M_\beta(\mathcal{D},\sigma)$. Assume that $K\ge1$ and $0<\kappa\le\gamma\le\frac{c_0}4$, verifying
\begin{equation}\label{criticalcondition}
\exp\lt\{8d\lt(\displaystyle\frac{\epsilon_0}\gamma\rt)^{\frac1{2\beta}}\rt\}\kappa\le\gamma.
\end{equation} Then there exists a subset $\mathcal{D}'=\mathcal{D}'(\gamma,\kappa,K)$, satisfying
\begin{equation}\label{estimateMprop}
\text{Meas}(\mathcal{D}\setminus\mathcal{D}')\le CK^{\alpha_1}\gamma^{\alpha_2}
\end{equation}
and $\mathcal{C}^1$ mappings $\widetilde{N}:\mathcal{D}'\mapsto\M_\beta\cap\nf,~ R:\mathcal{D}'\times\mathbb{T}_{\sigma'}^n\mapsto\M_\beta$ and $S:\mathcal{D}'\times\mathbb{T}_{\sigma'}^n\mapsto\M_\beta^+$, Hermitian and
analytic in $\varphi$, such that
\begin{equation}\label{homologicaleqnprop}
-\omega\cdot\nabla_{\varphi}S+\mi[N,S]=\widetilde{N}-Q+R
\end{equation}
and for any $0<\sigma'<\sigma$
\begin{align}
[\widetilde{N}]_\beta^{\mathcal{D}'}&\le[Q]_\beta^{\mathcal{D},\sigma},\label{estimateNprop}\\
[R]_\beta^{\mathcal{D}',\sigma'}&\le\frac{C\me^{-\frac{K}2(\sigma-\sigma')}}{(\sigma-\sigma')^n}[Q]_\beta^{\mathcal{D},\sigma},\label{estimateRprop}\\
[S]_{\beta+}^{\mathcal{D}',\sigma'}&\le\frac{CK\gamma}{\kappa^3(\sigma-\sigma')^n}[Q]_\beta^{\mathcal{D},\sigma},\label{estimateSprop}
\end{align}
where the constant $C>0$ depends on $n,d,\beta$ and $\alpha_1=\max\{\tau_1,n+d\},~\alpha_2=\max\{\tau_2,1\}$. 
\end{proposition}
\begin{proof}
Written in Fourier variables (w.r.t. $\varphi$), the homological equation \eqref{homologicaleqnprop} reads 
\begin{equation*}
-\mi k\cdot\omega\widehat{S}(k)+\mi[N,\widehat{S}(k)]=\delta_{k,0}\widetilde{N}-\widehat{Q}(k)+\widehat{R}(k),
\end{equation*}
where $\delta_{k,j}$ denotes the Kronecker symbol.

Decompose the equation on each product block $[a]\times[b]$:
\begin{equation}\label{homologicaleqnComponent}
L\widehat{S}_{[a]}^{[b]}(k)=\mi\delta_{k,0}\widetilde{N}_{[a]}^{[b]}-\mi\widehat{Q}_{[a]}^{[b]}(k)+\mi\widehat{R}_{[a]}^{[b]}(k),
\end{equation}
where $L:=L(k,[a],[b],\omega)$ is the linear operator, acting on the space of complex $[a]\times[b]$-matrices defined by
\begin{equation}\label{operatorL}
LM=(k\cdot\omega-N_{[a]}(\omega))M+MN_{[b]}(\omega)\text{ with }N_{[a]}=N_{[a]}^{[a]}.
\end{equation}

First solve this equation when $|k|+|w_a-w_b|=0$ (i.e. $k=0,~w_a=w_b$) by defining
\[
\widehat{S}_{[a]}^{[a]}(0)=0,~\widehat{R}_{[a]}^{[a]}(0)=0\text{ and }\widetilde{N}_{[a]}^{[a]}=\widehat{Q}_{[a]}^{[a]}(0).
\]
Then setting $\widetilde{N}_{[a]}^{[b]}=0$ for $w_a\ne w_b$  we obtain  $\widetilde{N}\in\M_\beta\cap\nf$ satisfies 
\(
|\widetilde{N}|_\beta\le|\widehat{Q}(0)|_\beta.
\) 
The estimates of the derivatives (w.r.t.\,$\omega$) are obtained by differentiating the expressions of $\widetilde{N}$. Taking all the estimates leads to \eqref{estimateNprop}.

Now turn to the other cases when $|k|+|w_a-w_b|>0$. Diagonalize the (Hermitian) matrix $N_{[a]}$ in an orthonormal basis:
$
\overline{P_{[a]}^T}N_{[a]}P_{[a]}=D_{[a]}
$
and denote $\widehat{S'}_{[a]}^{[b]}=\overline{P_{[a]}^T}\widehat{S}_{[a]}^{[b]}P_{[b]},\,\widehat{Q'}_{[a]}^{[b]}=\overline{P_{[a]}^T}\widehat{Q}_{[a]}^{[b]}P_{[b]}$ and $\widehat{R'}_{[a]}^{[b]}=\overline{P_{[a]}^T}\widehat{R}_{[a]}^{[b]}P_{[b]}$. Here we note for later use that $\|\widehat{M'}_{[a]}^{[b]}\|=\|\widehat{M}_{[a]}^{[b]}\|$ for $M=S,Q,R$. In this new variables the homological equation \eqref{homologicaleqnComponent} reads
\begin{equation*}
(k\cdot\omega-D_{[a]})\widehat{S'}_{[a]}^{[b]}(k)+\widehat{S'}_{[a]}^{[b]}(k)D_{[b]}=-\mi\widehat{Q'}_{[a]}^{[b]}(k)+\mi\widehat{R'}_{[a]}^{[b]}(k).
\end{equation*}
We solve it term by term: let $a,b\in\mathcal{E}$ and  set
\begin{alignat}{2}
&\hspace{20pt}\begin{alignedat}{4}\label{solutionR}
\widehat{R'}_{[a]}^{[b]}(k)&=0,&&\text{for }|k|\le K,\\
\widehat{R'}_{jl}(k)&=\widehat{Q'}_{jl}(k),~j\in[a],l\in[b],&\quad&\text{for }|k|>K
\end{alignedat}
\intertext{and}
&\begin{alignedat}{4}\label{solutionS}
\widehat{S'}_{[a]}^{[b]}(k)&=0,&\hspace{1.8cm}&\text{for }|k|>K\text{ or }|k|+|w_a-w_b|=0,\\
\left(\widehat{S'}_{[a]}^{[b]}(k)\right)_{jl}&=\frac{-\mi\left(\widehat{Q'}_{[a]}^{[b]}(k)\right)_{jl}}{k\cdot\omega-\alpha_j+\beta_l},&\quad&\text{in the other cases},
\end{alignedat}
\end{alignat}
where $\alpha_j(\omega)$ and $\beta_l(\omega)$ denote eigenvalues of $N_{[a]}(\omega)$ and $N_{[b]}(\omega)$, respectively. Before the estimations of such matrices, first remark the following assertion:
\[
\overline{\left(\widehat{Q'}_{[a]}^{[b]}(k)\right)_{jl}}=\left(\widehat{Q'}_{[b]}^{[a]}(-k)\right)_{lj}\Rightarrow\overline{\left(\widehat{S'}_{[a]}^{[b]}(k)\right)_{jl}}=\left(\widehat{S'}_{[b]}^{[a]}(-k)\right)_{lj}.
\] 
Hence, if $Q'$ verifies condition \eqref{hermitian}, so it is with $S'$ which implies the flow generated by $S$ preserves the symmetry $\eta=\bar{\xi}$.

Canonically, \eqref{solutionR} leads to
\begin{equation}\label{estimateRnorm}
|R(\varphi)|_\beta=|R'(\varphi)|_\beta\le\frac{Ce^{-\frac K2(\sigma-\sigma')}}{(\sigma-\sigma')^n}\sup_{|\Im\varphi|<\sigma}|Q(\varphi)|_\beta,\text{ for }|\Im\varphi|<\sigma'.
\end{equation}
Facing the small divisors
\begin{equation*}
k\cdot\omega-\alpha_j(\omega)+\beta_l(\omega),\quad j\in[a],l\in[b]\text{ and }[a],[b]\in\widehat{\mathcal{E}},
\end{equation*}
we distinguish two cases depending on whether $k=0$ or not.
\\[8pt]
\textbf{The case $k=0$.} In this case we know that $w_a\ne w_b$ which implies  $|w_a-w_b|\ge2$. Using \eqref{asymptotics} and \eqref{NcloseN0prop} we get that, if $\kappa\le\gamma\le\frac{c_0}4$ and $\epsilon_0\le\frac{c_0}4$,
\begin{alignat*}{2}
|-\lambda_a+\lambda_b|&\ge c_0|w_a-w_b|\ge\frac{c_0}2(1+|w_a-w_b|)\ge2 \gamma(1+|w_a-w_b|)
\intertext{and}
|-\alpha_j(\omega)+\beta_l(\omega)|&\ge|-\lambda_a+\lambda_b|-4\epsilon_0\ge\frac{c_0}2|w_a-w_b|\ge\kappa(1+|w_a-w_b|).
\end{alignat*}
Collecting the last two estimates and condition \eqref{NcloseN0prop} allows us to utilize Lemma \ref{criticallemma} to conclude that
\begin{equation}\label{estimateS0}
|\widehat{S}(0)|_{\beta+}\le\kappa^{-1}\exp\left\{d\left(\frac{\epsilon_0}\gamma\right)^{\frac1{2\beta}}\right\}|\widehat{Q}(0)|_\beta.
\end{equation}
\textbf{The case $k\ne0$.} Concretely, in this case we only solve the main terms of Fourier series truncated at order $K$. Utilizing Hypothesis H2, for any $\gamma>0$, there exists a subset $D_1=\mathcal D(2\gamma,K)$, satisfying
$
\text{Meas}(\mathcal{D}_0\setminus D_1)\le CK^{\tau_1}\gamma^{\tau_2},
$
such that for all $\omega\in D_1$ and $0<|k|\le K$, 
$
|k\cdot\omega-\lambda_a+\lambda_b|\ge2\gamma(1+|w_a-w_b|).
$
By \eqref{NcloseN0prop} this implies
\begin{align*}
|k\cdot\omega-\alpha_j(\omega)+\beta_l(\omega)|&\ge|k\cdot\omega-\lambda_a+\lambda_b|-|\alpha_j(\omega)-\lambda_a|-|\beta_l(\omega)-\lambda_b|\\
&\ge2\gamma(1+|w_a-w_b|)-\frac{2\epsilon_0}{(1+\ln w_a)^{2\beta}}-\frac{2\epsilon_0}{(1+\ln w_b)^{2\beta}}\\
&\ge\gamma(1+|w_a-w_b|),\qquad\underline{\text{if }w_b\ge w_a\ge\exp\left\{\left(\frac{4\epsilon_0}{\gamma}\right)^{\frac1{2\beta}}\right\}.}
\end{align*}
Now let $\displaystyle w_a\le\exp\left\{\left(\frac{4\epsilon_0}{\gamma}\right)^{\frac1{2\beta}}\right\}$. 
In fact if $|w_a-w_b|\geq CK$, we can prove that 
$$|k\cdot\omega-\alpha_j(\omega)+\beta_l(\omega)|\geq \kappa(1+|w_a-w_b|).$$ 
We face the case $|w_a-w_b|\le CK$ which follows $\displaystyle w_b\le CK\exp\left\{\left(\frac{4\epsilon_0}{\gamma}\right)^{\frac1{2\beta}}\right\}$. Since $|\partial_{\omega}(k\cdot\omega)(\frac k{|k|})|=|k|\ge1$, condition \eqref{NcloseN0prop} implies
$
\Big|\partial_{\omega}\big(k\cdot\omega-\alpha_j(\omega)+\beta_l(\omega)\big)\Big(\frac k{|k|}\Big)\Big|\ge\frac{|k|}2.
$
The last estimate allows us to use  Lemma \ref{classicallemmaMeasure} to conclude that
\begin{equation*}
|k\cdot\omega-\alpha_j(\omega)+\beta_l(\omega)|\ge\kappa(1+|w_a-w_b|),\quad\forall\,j\in[a],l\in[b]
\end{equation*}
except  a set $F_{[a],[b],k}$ whose measure is smaller than $Cw_a^{d-1}w_b^{d-1}\frac{\kappa(1+|w_a-w_b|)}{|k|}$. Denoting $F$ be the union of $F_{[a],[b],k}$ for $[a],[b]\in\widehat{\mathcal{E}}$ and $0<|k|\le K$ such that $\displaystyle w_a\le\exp\left\{\left(\frac{4\epsilon_0}{\gamma}\right)^{\frac1{2\beta}}\right\}$ and $\displaystyle w_b\le CK\exp\left\{\left(\frac{4\epsilon_0}{\gamma}\right)^{\frac1{2\beta}}\right\}$ with $|w_a-w_b|\le CK$, condition \eqref{criticalcondition} leads to
\[
	\text{Meas}(F)\le Cw_a^dw_b^dK^n\kappa\le CK^{n+d}\exp\left\{8d\left(\frac{\epsilon_0}{\gamma}\right)^{\frac1{2\beta}}\right\}\kappa
\le CK^{n+d}\gamma.
\]
Let $\mathcal{D}'=D_1\cup D_2$ with $D_2=\mathcal{D}\setminus F$ and $\alpha_1=\max\{\tau_1,n+d\}$ and $\alpha_2=\max\{\tau_2,1\}$, then
\[
	\text{Meas}(\mathcal{D}\setminus\mathcal{D}')\le\text{Meas}(\mathcal{D}_0\setminus D_1)+
\text{Meas}(F)\le CK^{\tau_1}\gamma^{\tau_2}+CK^{n+d}\gamma\le CK^{\alpha_1}\gamma^{\alpha_2}.
\]
Further, by construction, for all $\omega\in\mathcal{D}',0<|k|\le K,a,b\in\mathcal{E}$ and $j\in[a],l\in[b]$ we have
\[
|k\cdot\omega-\lambda_a+\lambda_b|\ge2\gamma(1+|w_a-w_b|)\text{ and }|k\cdot\omega-\alpha_j(\omega)+\beta_l(\omega)|\ge\kappa(1+|w_a-w_b|).
\]
Hence, in view of \eqref{solutionS}, utilizing Lemma \ref{criticallemma} concludes that $\widehat{S'}(k)\in\M_\beta^+$ satisfies
\[
|\widehat{S}(k)|_{\beta+}=|\widehat{S'}(k)|_{\beta+}\le\kappa^{-1}\exp\left\{d\left(\frac{\epsilon_0}\gamma\right)^{\frac1{2\beta}}\right\}|\widehat{Q}(k)|_\beta,\quad0<|k|\le K.
\]
Combining the last estimate with \eqref{estimateS0} we obtain a solution $S$ satisfying for all $|\Im\varphi|<\sigma'$
\begin{equation}\label{estimateSnorm}
|S(\varphi)|_{\beta+}\le \frac{C}{\kappa(\sigma-\sigma')^n}\exp\left\{d\left(\frac{\epsilon_0}\gamma\right)^{\frac1{2\beta}}\right\}\sup_{|\Im\varphi|<\sigma}|Q(\varphi)|_\beta.
\end{equation}
To obtain the estimates for the derivative (w.r.t.\,$\omega$) we differentiate \eqref{homologicaleqnComponent}:
\[
L(\partial_{\omega}\widehat{S}_{[a]}^{[b]}(k,\omega))=-(\partial_{\omega}L)\widehat{S}_{[a]}^{[b]}(k,\omega)-\mi\partial_{\omega}\widehat{Q}_{[a]}^{[b]}(k,\omega)+\mi\partial_{\omega}\widehat{R}_{[a]}^{[b]}(k,\omega)
\]
which is an equation of the same type as \eqref{homologicaleqnComponent} for $\partial_{\omega}\widehat{S}_{[a]}^{[b]}(k,\omega)$ and $\partial_{\omega}\widehat{R}_{[a]}^{[b]}(k,\omega)$ where $-\mi\widehat{Q}_{[a]}^{[b]}(k,\omega)$ is replaced by $B_{[a]}^{[b]}(k,\omega)=-(\partial_{\omega}L)\widehat{S}_{[a]}^{[b]}(k,\omega)-\mi\partial_{\omega}\widehat{Q}_{[a]}^{[b]}(k,\omega)$. Solve this equation by defining
\begin{align*}
\partial_{\omega}\widehat{S}_{[a]}^{[b]}(k,\omega)&=\chi_{|k|\le K}L^{-1}(k,[a],[b],\omega)B_{[a]}^{[b]}(k,\omega),\\
\partial_{\omega}\widehat{R}_{[a]}^{[b]}(k,\omega)&=\mi\chi_{|k|>K}B_{[a]}^{[b]}(k,\omega)=\chi_{|k|>K}\partial_{\omega}\widehat{Q}_{[a]}^{[b]}(k,\omega).
\end{align*}
Collecting condition \eqref{NcloseN0prop} and the definition \eqref{operatorL} leads to $|(\partial_{\omega}L)\widehat{S}(k,\omega)|_\beta\le CK|\widehat{S}(k,\omega)|_\beta$ which implies
\[
|B(k,\omega)|_\beta\le \frac{CK}{\kappa}\exp\left\{d\left(\frac{\epsilon_0}\gamma\right)^{\frac1{2\beta}}\right\}\left(|\widehat{Q}(k)|_\beta+|\partial_{\omega}\widehat{Q}(k)|_\beta\right).
\] 
Following the same strategy as in the resolution of \eqref{homologicaleqnComponent} we get for $|\Im\varphi|<\sigma'$ 
\begin{equation}\label{estimateSRlipnorm}
\hspace{-12pt}
\begin{aligned}
|\partial_{\omega}S(\varphi)|_{\beta+}&\le \frac{CK}{\kappa^2(\sigma-\sigma')^n}\exp\left\{2d\left(\frac{\epsilon_0}\gamma\right)^{\frac1{2\beta}}\right\}\left(\sup_{|\Im\varphi|<\sigma}|Q(\varphi)|_\beta+\sup_{|\Im\varphi|<\sigma}|\partial_{\omega}Q(\varphi)|_\beta\right),\\
|\partial_{\omega}R(\varphi)|_\beta&\le\frac{C\me^{-\frac K2(\sigma-\sigma')}}{(\sigma-\sigma')^n}\sup_{|\Im\varphi|<\sigma}|\partial_{\omega}Q(\varphi)|_\beta.
\end{aligned}
\end{equation}
Collecting \eqref{estimateRnorm},\eqref{estimateSnorm} and the last estimates  \eqref{estimateSRlipnorm} and taking into account \eqref{criticalcondition} leads to \eqref{estimateRprop} and \eqref{estimateSprop}.
\end{proof}

\subsection{The KAM iteration}
In this section we will prove Theorem \ref{reduthm} by an iterative KAM procedure. Let us begin with the initial Hamiltonian $H_\omega=h_0+q_0$ where
\begin{equation*}
h_0(y,\varphi,\xi,\eta)=\omega\cdot y+\langle\xi,N_0\eta\rangle,
\end{equation*}
$N_0=\text{diag}(\lambda_a:a\in\mathcal{E}),\omega\in\mathcal{D}_0$ and the quadratic perturbation $q_0^{}(\varphi,\xi,\eta)=\langle\xi,Q_0(\varphi)\eta\rangle$ with $Q_0=\epsilon Q\in\M_\beta(\mathcal{D}_0,\sigma_0)$ and $\sigma_0=\sigma$. Building iteratively the change of variables $\phi_{\chi_m^{}}^1$, we obtain the normal form $h_m=\omega\cdot y+\langle\xi,N_m(\omega)\eta\rangle$ and the perturbation $q_m^{}=\langle\xi,Q_m(\omega,\varphi)\eta\rangle$ with $Q_m\in\M_\beta(\mathcal{D}_m,\sigma_m)$ as follows: assume that the construction has been built up to step $m\ge0$ then 
\begin{enumerate}
	\item[(i).\phantom{i}] we utilize Proposition \ref{iterationproposition} to construct $S_{m+1}(\omega,\varphi)$ solution of the homological equation verifying for $(\omega,\varphi)\in\mathcal{D}_{m+1}\times\mathbb{T}_{\sigma_{m+1}}^n$
\begin{equation}\label{homologicaleqniteration}
-\omega\cdot\nabla_{\varphi}S_{m+1}+\mi[N_{m},S_{m+1}]=\widetilde{N}_{m}-Q_m+R_{m}
\end{equation}
where $\widetilde{N}_{m}(\omega),R_{m}(\omega,\varphi)$ defined for $(\omega,\varphi)\in\mathcal{D}_{m+1}\times\mathbb{T}_{\sigma_{m+1}}^n$ by
\begin{align}
\widetilde{N}_{m}(\omega)&=\left(\delta_{[j],[l]}\widehat{Q}_m(0)_{jl}\right)_{j,l\in\mathcal{E}},\label{solutionNtildeiteration}\\
R_{m}(\omega,\varphi)&=\sum_{|k|>K_{m+1}}\widehat{Q}_m(\omega,k)e^{\mi k\cdot\varphi};\label{solutionRiteration}
\end{align}
\item[(ii).] we define $N_{m+1},Q_{m+1}$ for $(\omega,\varphi)\in\mathcal{D}_{m+1}\times\mathbb{T}_{\sigma_{m+1}}^n$ by
\begin{align}
N_{m+1}&=N_m+\widetilde{N}_{m},\label{solutionNiteration}\\
Q_{m+1}&=R_{m}+\mi\int_0^1\me^{-\mi tS_{m+1}}[(1-t)(\widetilde{N}_{m}+R_{m})+tQ_m,S_{m+1}]\me^{\mi tS_{m+1}}\md t\label{solutionQiteration}.
\end{align}
\end{enumerate}
By construction, if $Q_m$ and $N_m$ are Hermitian, so it is with all of  $\widetilde{N}_m,R_{m}$ and $S_{m+1}$, by resolution of the homological equation, and also $N_{m+1}$ and $Q_{m+1}$. Let
\begin{equation}\label{HamiltonianIteration}
\begin{aligned}
h_{m+1}(y,\varphi,\xi,\eta;\omega)&=\omega\cdot y+\langle\xi,N_{m+1}(\omega)\eta\rangle,\\
\chi_{m+1}^{}(y,\varphi,\xi,\eta;\omega)&=\langle\xi,S_{m+1}(\omega,\varphi)\eta\rangle,\\
q_{m+1}^{}(y,\varphi,\xi,\eta;\omega)&=\langle\xi,Q_{m+1}(\omega,\varphi)\eta\rangle.
\end{aligned}
\end{equation}
Recall that $\phi_\chi^t$ denotes the time t flow generated by $S$ (see \eqref{change}), then
\begin{align*}
	f\circ\phi_{\chi_{m+1}^{}}^1&=f+\int_0^1\{f,\chi_{m+1}^{}\}\circ\phi_{\chi_{m+1}^{}}^t\md t
\intertext{or}
	f\circ\phi_{\chi_{m+1}^{}}^1&=f+\{f,\chi_{m+1}^{}\}+\int_0^1(1-t)\{\{f,\chi_{m+1}^{}\},\chi_{m+1}^{}\}\circ\phi_{\chi_{m+1}^{}}^t\md t.
\end{align*}
Therefore collecting \eqref{liebracket} and \eqref{change} leads to for $\omega\in\mathcal{D}_{m+1}$
\begin{alignat*}{5}
&\,(h_m+q_m^{})\circ\phi_{\chi_{m+1}^{}}^1=h_m\circ\phi_{\chi_{m+1}^{}}^1+q_m^{}\circ\phi_{\chi_{m+1}^{}}^1\\
=&\,h_m+\{h_m,\chi_{m+1}^{}\}+q_m^{}+\int_0^1\{(1-t)\{h_m,\chi_{m+1}^{}\}+q_m^{},\chi_{m+1}^{}\}\circ\phi_{\chi_{m+1}^{}}^t\md t\\
=&\,h_m+\langle\xi,(\widetilde{N}_{m}+R_{m})\eta\rangle+\int_0^1\{\langle\xi,\big((1-t)(\widetilde{N}_{m}+R_{m})+tQ_m\big)\eta\rangle,\chi_{m+1}^{}\}\circ\phi_{\chi_{m+1}^{}}^t\md t\\
:=&\,h_{m+1}+q_{m+1}^{}.
\end{alignat*}

\subsection{Iterative lemma}
Following the general iterative procedures  \eqref{homologicaleqniteration}-\eqref{HamiltonianIteration} we have 
\[
(h_0+q_0^{})\circ\phi_{\chi_1^{}}^1\circ\phi_{\chi_2^{}}^1\circ\cdots\circ\phi_{\chi_m^{}}^1=h_m+q_m^{}
\]
where $h_m=\omega\cdot y+\langle\xi,N_m\eta\rangle$ with $N_m\in\nf$ and $q_m^{}=\langle\xi,Q_m\eta\rangle$ with $Q_m\in\M_\beta(\mathcal{D}_m,\sigma_m)$. At the step $m$ the Fourier series are truncated at order $K_m$ and the small divisors are controlled by $\gamma_m$ and $\kappa_m$. Specifically, we choose all of the parameters for $m\ge0$ in term of $\epsilon_m$ which will control $[Q_m]_\beta^{\mathcal{D}_m,\sigma_m}$ as follows: 
first define $\sigma_0=\sigma$ and $\epsilon_0$ verifying $[Q_0]_\beta^{\mathcal{D}_0,\sigma_0}\le\epsilon_0$ and denote $\alpha=\frac{\alpha_1}{\alpha_2}+1$, then for $m\ge1$ let
\begin{equation}\label{iterativechoice}
\begin{gathered}
\epsilon_m=\epsilon_{m-1}^{5/4},\kappa_m=\epsilon_{m-1}^{1/4},\gamma_m=\epsilon_0^{1/6}(\ln\epsilon_{m-1}^{-1})^{-\alpha},\\
\sigma_{m-1}-\sigma_m=C_*\sigma_0m^{-2},K_m=2(\sigma_{m-1}-\sigma_m)^{-1}\ln\epsilon_{m-1}^{-1},
\end{gathered}
\end{equation}
where $(C_*)^{-1}=2\sum_{m\ge1}m^{-2}$ and $\alpha_1=\max\{\tau_1,n+d\},\alpha_2=\max\{\tau_2,1\}$.
\begin{remark}\label{criticalconditionholds}
From (\ref{iterativechoice}), the assumption \eqref{criticalcondition} in the KAM iteration is equivalent to 
\begin{equation}\label{equivcondition}
8d\epsilon_0^{\frac5{12\beta}}\left(\ln\epsilon_{m-1}^{-1}\right)^{\frac{\alpha}{2\beta}}+
\frac16\ln\epsilon_0^{-1}+\alpha\ln\big(\ln\epsilon_{m-1}^{-1}\big)\le\frac14\ln\epsilon_{m-1}^{-1},\quad\forall~m\ge1.
\end{equation}
Clearly, if $\beta\ge\frac\alpha2$,  \eqref{equivcondition} or  \eqref{criticalcondition} always holds true for $\epsilon_0\ll1$. 
\end{remark}
\begin{lemma}\label{iterationlemma}
Let $\alpha=\frac{\alpha_1}{\alpha_2}+1$ with $\alpha_1=\max\{\tau_1,n+d\},\alpha_2=\max\{\tau_2,1\}$ and $\beta\ge\frac{\alpha}{2}$. There exists $\epsilon_*$ depending on $\sigma, d,n,\beta,\tau_1,\tau_2$ and $h_0$ such that, for $0\le\epsilon_0<\epsilon_*$ and
$
\epsilon_m=\epsilon_0^{(5/4)^m},\quad m\ge0,
$
the followings hold for all $m\ge1$:
there exists $\mathcal{D}_m\subset\mathcal{D}_{m-1},S_m\in\M_\beta^+(\mathcal{D}_m,\sigma_m),h_m=\omega\cdot y+\langle\xi,N_m\eta\rangle$ with $N_m\in\nf$ and $Q_m\in\M_\beta(\mathcal{D}_m,\sigma_m)$ such  that
\begin{enumerate}
	\item[(i).\phantom{i}] for $p\in[-1,1]$ the transformation
\begin{equation}\label{definitionTransformationstepm}
\phi_m^{}(\cdot,\omega,\varphi):=\phi_{\chi_m^{}}^1:Y_p\mapsto Y_p,\quad\forall\,(\omega,\varphi)\in\mathcal{D}_m\times\mathbb{T}_{\sigma_m}^n
\end{equation}
is linear (unitary in $Y_0$) isomorphism conjugating the Hamiltonian at step $m-1$ to the Hamiltonian at step $m$, i.e.
\begin{equation*}
(h_{m-1}+q_{m-1}^{})\circ\phi_m^{}=h_m+q_m^{};
\end{equation*}
\item[(ii).] the following estimations hold:
\begin{gather}
\text{Meas}(\mathcal{D}_{m-1}\setminus\mathcal{D}_m)\le\epsilon_0^{\alpha_2/6}(\ln\epsilon_{m-1}^{-1})^{-\frac{\alpha_2}2},\label{iterationM}\\
[\widetilde{N}_{m-1}]_\beta^{\mathcal{D}_m}\le\epsilon_{m-1},\label{iterationN}\\
[Q_m]_\beta^{\mathcal{D}_m,\sigma_m}\le\epsilon_m\label{iterationQ},\\
[S_m]_{\beta+}^{\mathcal{D}_m,\sigma_m}\le\epsilon_0^{1/6}(\ln\epsilon_0^{-1})^{-\frac{\alpha_1}{2\alpha_2}}\epsilon_{m-1}^{1/4}\notag
\end{gather}
and for all $p\in[-1,1]$ the transformation satisfies
\begin{equation}\label{iterationtranform}
\|\phi_m^{}(\cdot,\omega,\varphi)-\Id\|_{\mathcal{B}(Y_p)}\le\epsilon_0^{1/6}(\ln\epsilon_0^{-1})^{-\frac{\alpha_1}{3\alpha_2}}\epsilon_{m-1}^{1/4},\quad\forall\,(\omega,\varphi)\in\mathcal{D}_m\times\mathbb{T}_{\sigma_m}^n.
\end{equation}
\end{enumerate} 
\end{lemma}
\begin{proof}
At step 1, the initial $h_0=\omega\cdot y+\langle\xi,N_0\eta\rangle$ and thus condition \eqref{NcloseN0prop} is trivially satisfied. Remark \ref{criticalconditionholds} allows us to apply Proposition \ref{iterationproposition} for constructing $S_1,\widetilde{N}_0,R_0$ and $\mathcal{D}_1,\sigma_1$ such that for $(\omega,\varphi)\in\mathcal{D}_1\times\mathbb{T}_{\sigma_1}^n$, 
$
-\omega\cdot\nabla_{\varphi}S_1+\mi[N_0,S_1]=\widetilde{N}_0-Q_0+R_0.
$
Then utilizing \eqref{estimateMprop}, we obtain for $\epsilon_0\ll1$
\[
\text{Meas}(\mathcal{D}_0\setminus\mathcal{D}_1)\le CK_1^{\alpha_1}\gamma_1^{\alpha_2}\le C\epsilon_0^{\alpha_2/6}(\ln\epsilon_0^{-1})^{-\alpha_2}
\le\epsilon_0^{\alpha_2/6}(\ln\epsilon_0^{-1})^{-\frac{\alpha_2}{2}}.
\]
Due to \eqref{estimateSprop} we have for $\epsilon_0\ll1$
\[
[S_1]_{\beta+}^{\mathcal{D}_1,\sigma_1}\le\frac{CK_1\gamma_1}{\kappa_1^3(\sigma_0-\sigma_1)^n}[Q_0]_\beta^{\mathcal{D}_0,\sigma_0}\le C\epsilon_0^{1/6}(\ln\epsilon_0^{-1})^{-\frac{\alpha_1}{\alpha_2}}\epsilon_0^{1/4}\le\epsilon_0^{1/6}(\ln\epsilon_0^{-1})^{-\frac{\alpha_1}{2\alpha_2}}\epsilon_0^{1/4}.
\]
Thus, in view of \eqref{change}, \eqref{definitionTransformationstepm} and assertion (iii),(v) of Lemma \ref{structure} we get for all $p\in[-1,1]$
\[
\|\phi_1^{}(\cdot,\omega,\varphi)-\Id\|_{\mathcal{B}(Y_p)}\le C\me^{C[S_1]_{\beta+}^{\mathcal{D}_1,\sigma_1}}[S_1]_{\beta+}^{\mathcal{D}_1,\sigma_1}\le\epsilon_0^{1/6}(\ln\epsilon_0^{-1})^{-\frac{\alpha_1}{3\alpha_2}}\epsilon_0^{1/4},\quad\text{if }\epsilon_0\ll1.
\]
Collecting \eqref{estimateNprop} and \eqref{estimateRprop} leads to $[\widetilde{N}_0]_\beta^{\mathcal{D}_1}\le\epsilon_0$ and for $\epsilon_0\ll1$
\[
[R_0]_\beta^{\mathcal{D}_1,\sigma_1}\le
\frac{C\me^{-\frac{K_1}2(\sigma_0-\sigma_1)}}{(\sigma_0-\sigma_1)^n}
[Q_0]_\beta^{\mathcal{D}_0,\sigma_0}\le C\epsilon_0^2\le\frac12\epsilon_0^{5/4}=\frac{\epsilon_1}{2}.
\]
Besides, \eqref{solutionQiteration} infers that for $\epsilon_0\ll1$
\[
[Q_1]_\beta^{\mathcal{D}_1,\sigma_1}\le[R_0]_\beta^{\mathcal{D}_1,\sigma_1}+C[Q_0]_\beta^{\mathcal{D}_0,\sigma_0}[S_1]_{\beta+}^{\mathcal{D}_1,\sigma_1}\le\frac{\epsilon_1}{2}+C\epsilon_0^{5/4+1/6}\le\epsilon_1.
\]
Now assume that we have verified Lemma \ref{iterationlemma} up to step $m$, then we consider the step $m+1$. Since $h_m=\omega\cdot y+\langle\xi,N_m\eta\rangle$ and
$
[N_m-N_0]_\beta^{\mathcal{D}_m}\le\sum_{l=0}^{m-1}[\widetilde{N}_{l}]_\beta^{\mathcal{D}_{l+1}}\le
\sum_{l=0}^{m-1}\epsilon_{l}\le2\epsilon_0,\ \text{if }\epsilon_0\ll1.
$
Then condition \eqref{NcloseN0prop} verifies and Remark \ref{criticalconditionholds} allows us to apply Proposition \ref{iterationproposition} for constructing $S_{m+1},\widetilde{N}_{m},R_{m}$ and $\mathcal{D}_{m+1},\sigma_{m+1}$ such that for $(\omega,\varphi)\in\mathcal{D}_{m+1}\times\mathbb{T}_{\sigma_{m+1}}^n$
\[
-\omega\cdot\nabla_{\varphi}S_{m+1}+\mi[N_m,S_{m+1}]=\widetilde{N}_{m}-Q_m+R_{m}.
\] 
Similarly, utilizing \eqref{estimateMprop} we obtain for $\epsilon_0\ll1$
\[
\text{Meas}(\mathcal{D}_m\setminus\mathcal{D}_{m+1})\le CK_{m+1}^{\alpha_1}\gamma_{m+1}^{\alpha_2}\le C\epsilon_0^{\alpha_2/6}\frac{(m+1)^{2\alpha_1}}{(\ln\epsilon_m^{-1})^{\alpha_2}}\le
\epsilon_0^{\alpha_2/6}(\ln\epsilon_m^{-1})^{-\frac{\alpha_2}{2}}.
\]
Due to \eqref{estimateSprop} we have for $\epsilon_0\ll1$
\[
[S_{m+1}]_{\beta+}^{\mathcal{D}_{m+1},\sigma_{m+1}}\le\frac{CK_{m+1}\gamma_{m+1}}{\kappa_{m+1}^3(\sigma_m-\sigma_{m+1})^n}[Q_m]_\beta^{\mathcal{D}_m,\sigma_m}
\le\epsilon_0^{1/6}(\ln\epsilon_0^{-1})^{-\frac{\alpha_1}{2\alpha_2}}\epsilon_m^{1/4}
\]
Thus, in view of \eqref{change},\eqref{definitionTransformationstepm} and assertion (iii),(v) of Lemma \ref{structure} we get for $p\in[-1,1]$
\[
\|\phi_{m+1}^{}(\cdot,\omega,\varphi)-\Id\|_{\mathcal{B}(Y_p)}\le C[S_{m+1}]_{\beta+}^{\mathcal{D}_{m+1},\sigma_{m+1}}\le\epsilon_0^{1/6}(\ln\epsilon_0^{-1})^{-\frac{\alpha_1}{3\alpha_2}}\epsilon_m^{1/4},\quad\text{if }\epsilon_0\ll1.
\]
Collecting \eqref{estimateNprop} and \eqref{estimateRprop} leads to $[\widetilde{N}_{m}]_\beta^{\mathcal{D}_{m+1}}\le\epsilon_m$ and
\[
[R_{m}]_\beta^{\mathcal{D}_{m+1},\sigma_{m+1}}\le
\frac{C\me^{-\frac{K_{m+1}}2(\sigma_m-\sigma_{m+1})}}{(\sigma_m-\sigma_{m+1})^n}[Q_m]_\beta^{\mathcal{D}_m,\sigma_m}
\le\frac{\epsilon_{m+1}}{2},\quad\text{if }\epsilon_0\ll1.
\]
In addition, \eqref{solutionQiteration} implies that for $\epsilon_0\ll1$
\begin{align*}
[Q_{m+1}]_\beta^{\mathcal{D}_{m+1},\sigma_{m+1}}&\le[R_{m}]_\beta^{\mathcal{D}_{m+1},\sigma_{m+1}}+C[Q_m]_\beta^{\mathcal{D}_m,\sigma_m}[S_{m+1}]_{\beta+}^{\mathcal{D}_{m+1},\sigma_{m+1}}\\
&\le\frac{\epsilon_{m+1}}{2}+C\epsilon_0^{1/6}(\ln\epsilon_0^{-1})^{-\frac{\alpha_1}{2\alpha_2}}\epsilon_m^{5/4}\le\epsilon_{m+1}.
\end{align*}
\end{proof}

\subsection{Transition to the limit and proof of reducibility theorem}
Let $\mathcal{D}_\epsilon=\cap_{m\ge0}\mathcal{D}_m$. In view of \eqref{iterationM}, this is a Borel set satisfying for $\epsilon_0\ll1$
\begin{align*}
\text{Meas}(\mathcal{D}_0\setminus\mathcal{D}_\epsilon)\le\sum_{m\ge0}\epsilon_0^{\alpha_2/6}(\ln\epsilon_m^{-1})^{-\frac{\alpha_2}{2}}=\epsilon_0^{\alpha_2/6}(\ln\epsilon_0^{-1})^{-\frac{\alpha_2}{2}}\sum_{m\ge0}\left((4/5)^{\frac{\alpha_2}2}\right)^m\le\epsilon_0^{\alpha_2/6}.
\end{align*}
This leads to the assertion (i) of Theorem \ref{reduthm}.  \\
\indent In the following, let $p\in[0,1],\,(\omega,\varphi)\in\mathcal{D}_\epsilon\times\mathbb{T}_{\sigma/2}^n$ and $\epsilon_0\ll1$. 
Collecting \eqref{iterationQ} and \eqref{iterationN} we conclude the direct lemmas as follows:
\begin{lemma}\label{limitQ}
$\big\{Q_m(\omega,\varphi)\big\}_{m\ge1}$ is a Cauchy sequence in $\M_\beta$ and
$
Q_m(\omega,\varphi)\to\boldsymbol{0}\text{ when }m\to\infty.
$
Furthermore, \eqref{iterationQ} infers the uniformly convergence on $(\omega,\varphi)$.
\end{lemma}
\begin{lemma}\label{limitN}
$\big\{N_m(\omega)-N_0\big\}_{m\ge1}$ is a Cauchy sequence in $\M_\beta$. Letting $W(\omega)\in\M_\beta$ be the limit mapping we have
$
N_m(\omega)-N_0\to W(\omega)\text{ when }m\to\infty.
$
Moreover, $\eqref{iterationN}$ implies the uniformly convergence on $\omega$, which leads to the $\mathcal{C}^1$  regularity.
\end{lemma}
To estimate the change of variables, we need the following two lemmas. 
\begin{lemma}
Let $\Phi_m=\phi_1^{}\circ\phi_2^{}\circ\cdots\circ\phi_m^{}$ for $m\ge1$, then we have
\begin{equation}\label{iterationTransform}
\|\Phi_m(\cdot,\omega,\varphi)-\Id\|_{\mathcal{B}(Y_p)}\le\epsilon_0^{1/6}(\ln\epsilon_0^{-1})^{-\frac{\alpha_1}{4\alpha_2}}\sum_{l=0}^{m-1}\epsilon_{l}^{1/4}.
\end{equation}
\end{lemma}
\begin{proof} 
First \eqref{iterationtranform} implies the above estimate \eqref{iterationTransform} holds true for $m=1$. Now assume that we have verified \eqref{iterationTransform} up to $m>1$, then we consider the case for $m+1$. From the definition,
\[
\Phi_{m+1}-\Id=\Phi_m\circ\phi_{m+1}-\Id=\Phi_m\circ(\phi_{m+1}-\Id)+\Phi_m-\Id.
\]
Therefore, collecting the assumption and \eqref{iterationtranform} leads to
\begin{align*}
	&\|\Phi_{m+1}(\cdot,\omega,\varphi)-\Id\|_{\mathcal{B}(Y_p)}\le\|\Phi_m(\cdot,\omega,\varphi)-\Id\|_{\mathcal{B}(Y_p)}+C\|\phi_{m+1}(\cdot,\omega,\varphi)-\Id\|_{\mathcal{B}(Y_p)}\\
	\le\,&\epsilon_0^{1/6}(\ln\epsilon_0^{-1})^{-\frac{\alpha_1}{4\alpha_2}}\sum_{l=0}^{m-1}\epsilon_{l}^{1/4}+C\epsilon_0^{1/6}(\ln\epsilon_0^{-1})^{-\frac{\alpha_1}{3\alpha_2}}\epsilon_m^{1/4}
	\le\epsilon_0^{1/6}(\ln\epsilon_0^{-1})^{-\frac{\alpha_1}{4\alpha_2}}\sum_{l=0}^{m}\epsilon_{l}^{1/4}.
\end{align*}
By induction, we complete the proof.
\end{proof}
\begin{lemma}
$\big\{\Phi_m(\cdot,\omega,\varphi)\big\}_{m\ge1}$ is a Cauchy sequence in $\mathcal{B}(Y_p)$. Letting $\Phi_\infty(\cdot,\omega,\varphi)\in\mathcal{B}(Y_p)$ be the limit mapping we have
$
\Phi_{m}(\cdot,\omega,\varphi)\to\Phi_\infty(\cdot,\omega,\varphi)\text{ when }m\to\infty.
$
Furthermore, \eqref{iterationTransform} implies the uniformly convergence on $(\omega,\varphi)$, which leads to $\Phi_\infty(\cdot,\omega,\varphi)$ is analytic in $\varphi$ and $\mathcal{C}^1$ in $\omega$. Moreover, \eqref{iterationTransform} implies
\begin{equation}\label{TransformationInfty}
	\|\Phi_\infty(\cdot,\omega,\varphi)-\Id\|_{\mathcal{B}(Y_p)}\le\epsilon_0^{5/12}.
\end{equation}
\end{lemma}
 \begin{proof}
From the definition we have
\[
\Phi_{m+1}-\Phi_m=\Phi_m\circ\phi_{m+1}^{}-\Phi_m=\Phi_m\circ(\phi_{m+1}^{}-\Id).
\]
Collecting \eqref{iterationtranform} and \eqref{iterationTransform} leads to 
\[
\|\Phi_{m+1}(\cdot,\omega,\varphi)-\Phi_m(\cdot,\omega,\varphi)\|_{\mathcal{B}(Y_p)}\le C\|\phi_{m+1}^{}(\cdot,\omega,\varphi)-\Id\|_{\mathcal{B}(Y_p)}
\le\epsilon_m^{1/4}.
\]
Hence, given $m_2^{}\ge m_1^{}\ge1$ we have 
\[
\|\Phi_{m_2^{}}(\cdot,\omega,\varphi)-\Phi_{m_1^{}}(\cdot,\omega,\varphi)\|_{\mathcal{B}(Y_p)}\le\sum_{l=m_1^{}}^{m_2^{}-1}\epsilon_{l}^{1/4}\le2\epsilon_{m_1}\to0\text{ when }m_1^{}\to\infty.
\]
The last estimate implies the Cauchy sequence, which concludes the results.
 \end{proof}

Define $ \mathcal{B}(\ell_p^2)\ni M_m(\omega,\varphi)=e^{\mi S_1(\omega,\varphi)}\circ e^{\mi S_2(\omega,\varphi)}\circ\cdots\circ e^{\mi S_m(\omega,\varphi)}$ for $m\ge1$. In view of \eqref{change} we get that
\begin{align}
\phi_m^{}(\xi,\eta,\omega,\varphi)&=\big(e^{-\mi S_m^T(\omega,\varphi)}\xi,e^{\mi S_m(\omega,\varphi)}\eta\big)\notag
\intertext{and}
\Phi_m(\xi,\eta,\omega,\varphi)&=\big(\overline{M_m(\omega,\varphi)}\xi,M_m(\omega,\varphi)\eta\big).\label{definitionPhim}
\end{align}
From a straightforward computation, given $m_2^{}\ge m_1^{}\ge1$
\[
\|M_{m_2^{}}(\omega,\varphi)-M_{m_1^{}}(\omega,\varphi)\|_{\mathcal{B}(\ell_p^2)}\le\|\Phi_{m_2^{}}(\cdot,\omega,\varphi)-\Phi_{m_1^{}}(\cdot,\omega,\varphi)\|_{\mathcal{B}(Y_p)}.
\]
This implies that $\big\{M_m(\omega,\varphi)\big\}_{m\ge1}$ is a Cauchy sequence in $\mathcal{B}(\ell_p^2)$. Letting $M_\infty(\omega,\varphi)$ be the limit mapping, the uniformly convergence leads to that $(\omega,\varphi)\mapsto M_\infty(\omega,\varphi)$ is analytic in $\varphi$ and $\mathcal{C}^1$ in $\omega$. Moreover, due to \eqref{definitionPhim}
\begin{equation}\label{PhiMInf}
	\Phi_\infty(\xi,\eta,\omega,\varphi)=\big(\overline{M_\infty(\omega,\varphi)}\xi,M_\infty(\omega,\varphi)\eta\big)
\end{equation}
and taking into account \eqref{TransformationInfty} leads to 
\begin{equation}\label{limitM}
\|\overline{M_\infty(\omega,\varphi)}-\Id\|_{\mathcal{B}(\ell_p^2)}=\|M_\infty(\omega,\varphi)-\Id\|_{\mathcal{B}(\ell_p^2)}\leq \|\Phi_\infty(\cdot,\omega,\varphi)-\Id\|_{\mathcal{B}(Y_p)}\le\epsilon_0^{5/12}.
\end{equation}
By construction the map $\Phi_m(\cdot,\omega,\omega t)$ conjugates the original Hamiltonian system associated with
$
H_\omega(t,\xi,\eta)=\langle\xi,N_0\eta\rangle+\epsilon\langle\xi,Q(\omega t)\eta\rangle
$
into the Hamiltonian system associated with 
$
H_m(t,\xi,\eta)=\langle\xi,N_m(\omega)\eta\rangle+\langle\xi,Q_m(\omega,\omega t)\eta\rangle.
$
Collecting Lemma \ref{limitQ} and \ref{limitN} one concludes $Q_m(\omega,\omega t)\to0$  and  $N_m(\omega)\to N_\omega$ when $m\to\infty$, where the operator $N_\omega\equiv N_\infty(\omega)=N_0+W(\omega)\in\nf$ 
is $\mathcal{C}^1$ in $\omega$ with
\begin{equation}\label{limitW}
	|W|_\beta=|N_\omega-N_0|_\beta\le2\epsilon_0.
\end{equation}

Let $\Phi_\omega(\varphi):=\Phi_\infty(\cdot,\omega,\varphi)$ and $M_\omega(\varphi)=M_\infty(\omega,\varphi)$, then \eqref{PhiMInf} reads
$
\Phi_\omega(\varphi)(\xi,\eta)=\big(\overline{M_\omega(\varphi)}\xi,M_\omega(\varphi)\eta\big).
$
Furthermore, denoting the limiting Hamiltonian $\mathcal{H}_\omega=\langle\xi,N_\omega\eta\rangle=\langle\xi,N_0\eta\rangle+\langle\xi,W\eta\rangle$, the symplectic coordinate transformation $\Phi_\omega(\varphi)$ conjugates the original Hamiltonian system associated with $H_\omega$ into the autonomous Hamiltonian system associated with $\mathcal{H}_\omega$. Collecting \eqref{limitM} and \eqref{limitW} leads to \eqref{NcloseN0} and \eqref{Mcloseidentity} in Theorem \ref{reduthm}.\qed 

\section{Appendix}
\subsection{Proof of Lemma \ref{structure}}\label{proofstructurelemma}
 Recall that $\beta>\frac12$.
\begin{enumerate}
\item[(i).\phantom{ii}] 
The proof results from Lemma \ref{auxiliarylemma} with $\beta>\frac12$ and 
\begin{equation}\label{anxiliaryseries}
\sum_{c\in\widehat{\mathcal{E}}}\frac1{(1+\ln w_c)^{2\beta}(1+|w_b-w_c|)}\le C.
\end{equation}
\item[(ii).\phantom{i}] Similarly,  collecting \eqref{anxiliaryseries} and $\beta>\frac12$ conludes the results.
\item[(iii).] Use assertion (ii) of Lemma \ref{structure}.
\item[(iv).] Let $A\in\M_\beta$ and $s\ge1$. Then for any $\xi\in\ell_s^2$, we have 
\begin{align*}
\|A\xi\|_{-s}^2&=\sum_{a\in\widehat{\mathcal{E}}}w_a^{-s}\|\sum_{b\in\widehat{\mathcal{E}}}A_{[a]}^{[b]}\xi_{[b]}\|^2 \le\sum_{a\in\widehat{\mathcal{E}}} w_a^{-s} \left(\sum_{b\in\widehat{\mathcal{E}}}\|A_{[a]}^{[b]}\|\cdot\|\xi_{[b]}\|\right)^2\\
&\le|A|_\beta^2\sum_{a\in\widehat{\mathcal{E}}}w_a^{-s}\left(\sum_{b\in\widehat{\mathcal{E}}}\frac{w_b^{s/2}\|\xi_{[b]}\|}{(1+\ln w_a)^\beta(1+\ln w_b)^\beta w_b^{s/2}}\right)^2\\
&\le|A|_\beta^2\sum_{a\in\widehat{\mathcal{E}}}\frac1{(1+\ln w_a)^{2\beta}w_a^s}\left(\sum_{b\in\widehat{\mathcal{E}}}\frac1{(1+\ln w_b)^{2\beta}w_b^s}\right)\left(\sum_{b\in\widehat{\mathcal{E}}}w_b^s\|\xi_{[b]}\|^2\right)\\
&\le|A|_\beta^2\sum_{a\in\widehat{\mathcal{E}}}\frac1{(1+\ln w_a)^{2\beta}w_a}\left(\sum_{b\in\widehat{\mathcal{E}}}\frac1{(1+\ln w_b)^{2\beta}w_b}\right)\|\xi\|_s^2\\
&\le C^2|A|_\beta^2\|\xi\|_s^2.
\end{align*}
\item[(v).\phantom{i}] Case 1: $s\in[0,1]$. In this case   we first prove 
\begin{equation}\label{anxiliaryseriesI}
(I):=\sum_{b\in\widehat{\mathcal{E}}}\frac{(w_a/w_b)^s}{(1+\ln w_b)^{2\beta}(1+|w_a-w_b|)}\le C, \quad\forall\,s\in[0,1]\text{ and }\beta>\frac12.
\end{equation}
We split the series above into two parts as follows:
$$
(I)=\left(\sum_{w_b>\frac{w_a}2}+\sum_{w_b\le\frac{wa}2}\right)\frac{(w_a/w_b)^s}{(1+\ln w_b)^{2\beta}(1+|w_a-w_b|)}:=(I_1)+(I_2).
$$
For the former, $(w_a/w_b)^s\le2^s\le2$. Thus
\begin{equation}\label{anxiliaryseriesI1}
(I_1)\le\sum_{b\in\widehat{\mathcal{E}}}\frac2{(1+\ln w_b)^{2\beta}(1+|w_a-w_b|)}\le C.
\end{equation}
Then turn to the latter. Since $w_b\le\frac{w_a}2$, then $1+|w_a-w_b|\ge w_a-w_b\ge\frac{w_a}2\ge w_b$. Thus 
\(
1+|w_a-w_b|=(1+|w_a-w_b|)^s(1+|w_a-w_b|)^{1-s}\ge(w_a/2)^sw_b^{1-s}\ge\frac12w_a^sw_b^{1-s}.
\)
Therefore, one obtains
$
(I_2)\le\sum_{b\in\widehat{\mathcal{E}}}\frac2{(1+\ln w_b)^{2\beta}w_b}\le C.
$
Collecting the last estimate and \eqref{anxiliaryseriesI1} leads to the results \eqref{anxiliaryseriesI}.\\
\indent Now, we prepare to prove the assertion of (v) when $s\in [0,1]$. 
Since $A\in\M_\beta^+$, then for any $\xi\in\ell_s^2$,
\begin{alignat*}{5}
\hspace{32pt}&&\hspace{1pt}&\|A\xi\|_s^2\le\sum_{a\in\widehat{\mathcal{E}}}w_a^s\left(\sum_{b\in\widehat{\mathcal{E}}}\|A_{[a]}^{[b]}\|\cdot\|\xi_{[b]}\|\right)^2\\
&\le&&\sum_{a\in\widehat{\mathcal{E}}}\frac{|A|_{\beta+}^2}{(1+\ln w_a)^{2\beta}}\left(\sum_{b\in\widehat{\mathcal{E}}}\frac{(w_a/w_b)^{s/2}}{(1+\ln w_b)^\beta(1+|w_a-w_b|)^{1/2}}\frac{w_b^{s/2}\|\xi_{[b]}\|}{(1+|w_a-w_b|)^{1/2}}\right)^2\\
&\le&&\sum_{a\in\widehat{\mathcal{E}}}\frac{|A|_{\beta+}^2}{(1+\ln w_a)^{2\beta}}\left(\sum_{b\in\widehat{\mathcal{E}}}\frac{(w_a/w_b)^s}{(1+\ln w_b)^{2\beta}(1+|w_a-w_b|)}\right)\left(\sum_{b\in\widehat{\mathcal{E}}}\frac{w_b^s\|\xi_{[b]}\|^2}{1+|w_a-w_b|}\right)\\
&\le&&C|A|_{\beta+}^2\sum_{b\in\widehat{\mathcal{E}}}w_b^s\|\xi_{[b]}\|^2\sum_{a\in\widehat{\mathcal{E}}}\frac1{(1+\ln w_a)^{2\beta}(1+|w_a-w_b|)}\quad\underline{~\text{by \eqref{anxiliaryseriesI}}~}\\
&\le&&C^2|A|_{\beta+}^2\|\xi\|_{s}^2. \quad\underline{\text{ by Lemma \ref{auxiliarylemma} }~}
\end{alignat*}

Next turn to the other case: $s\in[-1,0)$. Repeating similar procedures as the first case and noting that
\begin{equation*}
\sum_{a\in\widehat{\mathcal{E}}}\frac{(w_a/w_b)^s}{(1+\ln w_a)^{2\beta}(1+|w_a-w_b|)}\le C,\quad \forall\,s\in[-1,0)\text{ and }\beta>\frac12, 
\end{equation*}
 we complete the proof.
\end{enumerate}\qed
\subsection{Some auxiliary lemmas}
\begin{lemma}[see Lemma A1 in \cite{WL2017}]\label{auxiliarylemma}
	For $j\ge1$ and $\delta>1$, there exists a positive constant $C\equiv C(\delta)$ independent of $j$ such that
	$
	\sum_{l\ge1}\frac{1}{(1+\ln l)^\delta(1+|l-j|)}\le C.
	$
\end{lemma}
The following lemma is classical. 
\begin{lemma}\label{classicallemmaMeasure}
	Let $f:[0,1]\mapsto\mathbb{R}$ be a $\mathcal{C}^1$ map satisfying $|f'(x)|\ge\delta$ for all $x\in[0,1]$ and let $\kappa>0$ then
	$
	\text{Meas}\big(\{x\in[0,1]:|f(x)|\le\kappa\}\big)\le\frac{2\kappa}{\delta}.
	$
\end{lemma}
\noindent
\textbf{Acknowledgements}\addcontentsline{toc}{section}{Acknowledgements}
The authors would like to thank the anonymous referee for corrections and suggestions. 
They were partially supported by NSFC grant (12071083) and Natural Science Foundation of Shanghai (19ZR1402400).


\begin{thebibliography}{99}

\bibitem{BBM2014}
Baldi, P.,  Berti, M., Montalto, R.: KAM for quasi-linear and fully nonlinear forced perturbations of Airy equation. Math. Ann.  \textbf{359}, 471-536 (2014)

\bibitem{BaMon21}
Baldi, P., Montalto, R.: Quasi - periodic incompressible Euler flows in 3D. Advances in Mathematics. \textbf{384}, 107730 (2021)

\bibitem{Bam2017}
Bambusi, D.:  Reducibility of 1-d Schr\"odinger equation with time quasiperiodic unbounded perturbations. II. Commun. Math. Phys. \textbf{353}, 353-378 (2017)

\bibitem{Bam2018}
Bambusi, D.:  Reducibility of 1-d Schr\"odinger equation with time quasiperiodic unbounded perturbations. I. Trans. Amer. Math. Soc. \textbf{370}, 1823-1865 (2018)


\bibitem{BLM18}
Bambusi, D., Langella, D., Montalto, R.: Reducibility of non-resonant transport equation on with unbounded perturbations. Ann. Henri Poincar\'e,  {\bf 20}, 1893-1929 (2019).



\bibitem{BLM2021}
Bambusi, D., Langella, D., Montalto, R.: Growth of Sobolev norms for unbounded perturbations of the Laplacian on flat tori. arXiv:2012.02654.



 \bibitem{BG2001}
Bambusi, D., Graffi, S.: Time quasi-periodic unbounded perturbations of Schr\"odinger operators and KAM methods. Commun. Math. Phys. \textbf{219}, 465-480 (2001)


\bibitem{BGMR2019}
Bambusi, D., Gr\'ebert, B., Maspero, A., Robert, D.: Growth of Sobolev norms for abstract linear Schr\"odinger equations.
J. Eur. Math. Soc. \textbf{23}, 557-583 (2021)

\bibitem{BGMR2018}
Bambusi, D., Gr\'{e}bert, B., Maspero, A., Robert, D.:  Reducibility of the quantum harmonic oscillator in $d$-dimensions with polynomial time-dependent perturbation. Anal. \& PDE \textbf{11}, 775-799 (2018)


\bibitem{BM2019}
Berti, M., Maspero, A.: Long time dynamics of Schr\"odinger and wave equations on flat tori.  J. Diff. Eqs., {\bf 267(2)}, 1167-1200 (2019). 



\bibitem{BM2016}
 Berti, M., Montalto, R.: Quasi-periodic standing wave solutions for gravity-capillary water waves. Memoirs of the American Mathematical Society, Volume \textbf{263}, Number 1273, 2020.
 
 
\bibitem{Bou99a}
Bourgain, J.: Growth of Sobolev norms in linear Schr\"odinger equations with smooth time dependent potentials. J. Anal. Math., {\bf 77}, 315-348 (1999). 

\bibitem{Bou99b}
Bourgain, J.: Growth of Sobolev norms in linear Schr\"odinger equations with quasi-periodic potential. Commun. Math. Phys., {\bf 204(1)}, 207-247 (1999).

 
 
 \bibitem{CoMon18} 
Corsi, L., Montalto, R.: Quasi-periodic solutions for the forced Kirchhoff equation on $\mathbb{T}^d$. Nonlinearity, \textbf{31}, 5075- 5109 (2018)
  
 
\bibitem{Com87}
Combescure, M.: The quantum stability problem for time-periodic perturbations of the harmonic oscillator. Ann. Inst. H. Poincar\'e Phys. Th\'eor. \textbf{47}(1), 63-83 (1987); Erratum: Ann. Inst. H. Poincar\'e Phys. Th\'eor. \textbf{47}(4), 451-454 (1987)


\bibitem{Del2014}
Delort, J.-M.: Growth of Sobolev norms for solutions of time dependent Schr\"odinger operators with harmonic oscillator potential. Commun. PDE \textbf{39}, 1-33 (2014)

\bibitem{DeSz2004}
Delort, J.-M., Szeftel, J.: Long-time existence for small data nonlinear Klein - Gordon equations on tori and spheres. Int. Math. Res. Not.  \textbf{37}, 1897-1966 (2004)

\bibitem{Eli1992}
Eliasson, L.H.: Floquet solutions for the 1-dimensional quasi-periodic Schr\"odinger equation. Commun. Math. Phys. \textbf{146}, 447-482 (1992)

\bibitem{EK2009} Eliasson, H.L., Kuksin, S. B.: On reducibility of Schr\"odinger equations with quasiperiodic in time potentials. Commun.  Math. Phys. \textbf{286}, 125-135 (2009) 

\bibitem{FZ12}
Fang, D., Zhang, Q.: On growth of Sobolev norms in linear Schr\"odinger equations with time dependent Gevrey potentials. J. Dynam. Differential Equations., {\bf 24(2)}, 151-180 (2012).

\bibitem{FaRa2020} Faou, E., Rapha\"el, P.: On weakly turbulent solutions to the perturbed linear harmonic oscillator.  arXiv: 2006.08206 (2020)

\bibitem{FGiMP19}
Feola, R., Giuliani, F., Montalto, R., Procesi, M.: Reducibility of first order linear operators on tori via Moser's theorem. J. Funct. Anal., {\bf 276(3)}, 932-970 (2019).

\bibitem{FGr19}
Feola, R., Gr\'ebert, B.: Reducibility of Schr\"odinger equation on the sphere. Int. Math. Res. Not. \textbf{0}, 1-39 (2020)


\bibitem{FGN19}
Feola, R., Gr\'ebert, B., Nguyen, T.: Reducibility of Schr\"odinger equation on a Zoll manifold with unbounded potential. J. Math. Phys. \textbf{61} (7), 071501 (2020) 

\bibitem{FP2015}
Feola, R., Procesi, M.: Quasi-periodic solutions for fully nonlinear forced reversible Schr\"odinger equations. J. Differ. Equ. \textbf{259}, 3389-3447 (2015)

\bibitem{GrYa2000}
Graffi, S., Yajima, K.: Absolute continuity of the Floquet spectrum for a nonlinearly forced harmonic oscillator. Commun. Math. Phys.  \textbf{215} (2), 245-250 (2000)


\bibitem{GP2016}
 Gr\'ebert, B., Paturel, E.: KAM for the Klein Gordon equation on $\mathbb S^d$. Boll. Unione Mat. Ital. \textbf{9}, 237-288 (2016)


\bibitem{GP2019}
 Gr\'{e}bert, B.,  Paturel, E.: On reducibility of quantum harmonic oscillator on $\mathbb R^d$ with quasiperiodic in time potential. 
 Annales de la Facult\'e des sciences de Toulouse : Math\'ematiques. \textbf{28}, 977-1014 (2019) 

\bibitem{GT2011}
  Gr\'{e}bert, B., Thomann,  L.: KAM for the quantum harmonic oscillator. Commun. Math. Phys. \textbf{307}, 383-427 (2011)

\bibitem{KT2005} Koch, H., Tataru, D.: $L^p$ eigenfunction bounds for the Hermite operator. Duke Math. J. \textbf{128}, 369-392 (2005)


 \bibitem{LiangLuo2021}
Liang, Z., Luo, J.: Reducibility of 1-d quantum harmonic oscillator equation with  unbounded oscillation perturbations.  J. Differ. Equ.  \textbf{270}, 343-389 (2021)
 
\bibitem{LiangW2019}
Liang, Z. , Wang, Z.: Reducibility of quantum harmonic oscillator on $\mathbb R^d$ with differential and quasi-periodic in time potential. J. Differ. Equ.  \textbf{267},  3355-3395 (2019)

\bibitem{LW2020}
Liang, Z., Wang, Z.Q.: Reducibility of 1-d Schr\"odinger equation with unbounded oscillation perturbations.  arXiv: 2003.13022v3 (2020)


\bibitem{LZZ2020}
Liang, Z., Zhao, Z., Zhou, Q.: 1-d quasi-periodic quantum harmonic oscillator with quadratic time-dependent perturbations: Reducibility and growth of Sobolev norms. 
J. Math. Pures Appl.  \textbf{146} 158-182 (2021)

\bibitem{LY2010}
  Liu, J., Yuan, X.: Spectrum for quantum duffing oscillator and small-divisor equation with large-variable coefficient. Commun. Pure Appl. Math. \textbf{63}, 1145-1172 (2010)
  
\bibitem{Mas2018}
Maspero, A.: Lower bounds on the growth of Sobolev norms in some linear time dependent Schr\"odinger equations. Math. Res. Lett. \textbf{26}, 1197-1215 (2019)

\bibitem{MR2017}
Maspero, A., Robert, D.: On time dependent Schr\"odinger equations: Global well-posedness and growth of Sobolev norms. J. Func. Anal., {\bf 273(2)}, 721-781 (2017).
  

\bibitem{Mon2014} Montalto, R.:  KAM for quasi-linear and fully nonlinear perturbations of Airy and KdV equations. Phd Thesis, SISSA - ISAS, 2014. 

\bibitem{Mon19}
Montalto, R.: A reducibility result for a class of linear wave equations on $\mathbb T^d$. Int. Math. Res. Notices, {\bf 2019(6)}, 1788-1862 (2019).


 \bibitem{PT2001}
Plotnikov, P.I., Toland, J.F.:  Nash-Moser theory for standing water waves. Arch. Rational Mech. Anal. \textbf{159}, 1-83 (2001)

\bibitem{ScTh2020}
Schwinte, V., Thomann, L.: Growth of Sobolev norms for coupled Lowest Landau Level equations. Pure Appl. Anal. \textbf{3}, 189-222 (2021)

 \bibitem{Th2020}
 Thomann, L.: Growth of Sobolev norms for linear Schr\"odinger operators. To appear in Pure Appl. Anal.  arXiv: 2006.02674 (2020)


\bibitem{WL2017} Wang, Z., Liang, Z.: Reducibility of 1D quantum harmonic oscillator perturbed by a quasiperiodic potential with logarithmic decay. Nonlinearity. \textbf{30}, 1405-1448 (2017)


\bibitem{Wang2008}
Wang, W.-M.: Pure point spectrum of the Floquet Hamiltonian for the quantum harmonic oscillator under time quasi-periodic perturbations. Commun. Math. Phys. \textbf{277},
459-496 (2008)

\bibitem{WWM2008}
Wang, W.-M.: Logarithmic bounds on Sobolev norms for time dependent linear Schr\"odinger equations.  Comm. Partial Differential Equations, {\bf 33(12)}, 2164-2179 (2008).






\end{thebibliography}
\end{document}